\documentclass[prodmode]{acmsmall} 

\usepackage[square,sort,comma,numbers]{natbib}

\usepackage{balance}
\usepackage{cancel,amsmath,float,balance,afterpage} 
\usepackage{multicol,rotating}

\usepackage{courier}            
\usepackage[scaled]{helvet} 
\usepackage{url}                  
\usepackage{listings}          
\usepackage{enumitem}      
\usepackage[colorlinks=true,allcolors=blue,breaklinks,draft=false]{hyperref}   

\usepackage{comment} 
\usepackage{tikz}
\usepackage{lipsum}
\usepackage{xspace}
\usepackage{amssymb}

\usepackage{amsmath}
\usepackage{comment}
\usepackage[lined,boxed,commentsnumbered, ruled,vlined, resetcount]{algorithm2e}
\usepackage{algpascal}
\usepackage{algc}
\usepackage{algcompatible}
\usepackage{algpseudocode}
\usepackage{listings}
\usepackage{xcolor, xspace}
\usepackage{bm}

\usepackage{caption}
\usepackage{setspace}
\usepackage{color,verbatim,mathabx}
\usepackage{graphicx}
\usepackage{listings}
\usepackage{multirow}
\usepackage{rotating}
\usepackage{multicol}
\usepackage{relsize}
\usepackage{pgfplots}
\usetikzlibrary{fit}
\usetikzlibrary{calc}
\usetikzlibrary{positioning}
\usetikzlibrary{shapes}
\usetikzlibrary{arrows}
\usepgflibrary{arrows}
\usepackage{graphicx}
\usepackage{xifthen}
\usepackage{pgfplotstable}
\usepackage{textcomp}

\definecolor{pink}  {rgb}{0.67, 0.05, 0.57} 
\definecolor{red}   {rgb}{0.87, 0.20, 0.00} 
\definecolor{green} {rgb}{0.00, 0.47, 0.00} 
\definecolor{violet}{rgb}{0.41, 0.12, 0.61} 
\definecolor{brown} {rgb}{0.39, 0.22, 0.13} 

\definecolor{Brown}{cmyk}{0,0.81,1,0.60}
\definecolor{OliveGreen}{cmyk}{0.64,0,0.95,0.40}
\definecolor{CadetBlue}{cmyk}{0.62,0.57,0.23,0}
\definecolor{lightlightgray}{gray}{0.9}

\newcommand{\bgt}{\eta}

\newcommand*\circled[1]{\tikz[baseline=(char.base)]{
	\node[shape=circle,draw,inner sep=0pt,fill={green!20!white}] (char) {#1};}}
\newcommand{\Cross}{$\mathbin{\tikz [x=1.3ex,y=1.3ex,line width=.2ex, red] \draw (0,0) -- (.8,.8) (0,.8) -- (.8,0);}$}%

\newcommand{\killset}{\textsl{kill}}
\newcommand{\singletons}{\textsl{singletons}}
\newcommand{\cxtsingletons}{\textsl{cxtSingletons}}

\newcommand{\svf}{\textsc{SVF}\xspace}
\newcommand{\vfsu}{\textsc{Supa}\xspace}
\newcommand{\supa}{\vfsu}
\newcommand{\vfsuFS}{\textsc{Supa-FS}\xspace}
\newcommand{\vfsuFSCS}{\textsc{Supa-FSCS}\xspace}
\newcommand{\fs}{\fontsize{8}{9}\selectfont}
\newcommand{\SFS}{\mbox{\textsc{SFS}}\xspace}
\newcommand{\FSCS}{\mbox{\textsc{FSCS}}\xspace}
\newcommand{\FS}{\textsc{FS}\xspace}
\newcommand{\FI}{\mbox{\textsc{FI}}\xspace}
\renewcommand{\footnotesize}{\fs}

\newcommand{\pts}{\emph{pt}}
\newcommand{\apts}{\emph{AnderPts}}
\newcommand{\aptg}{\emph{ptg}}

\newcommand{\allocation}[1]{\widehat{#1}\xspace}

\newcommand{\lab}{\ell\xspace}

\newcommand{\cxt}{c\xspace}
\newcommand{\cc}{\kappa\xspace}

\newcommand{\cpush}{\oplus\xspace}
\newcommand{\cpop}{\ominus\xspace}

\newcommand{\vfedge}[1]{\xrightarrow{#1}}

\newcommand{\calQ}{\mathcal{Q}\xspace}
\newcommand{\calI}{\mathcal{I}\xspace}

\newcommand{\calC}{\mathcal{C}\xspace}

\newcommand{\calG}{\mathcal{G}\xspace}

\newcommand{\calV}{\mathcal{V}\xspace}
\newcommand{\calL}{\mathcal{L}\xspace}

\newcommand{\calA}{\mathcal{A}\xspace}
\newcommand{\calS}{\mathcal{S}\xspace}
\newcommand{\calF}{\mathcal{F}\xspace}
\newcommand{\calP}{\mathcal{P}\xspace}
\newcommand{\calO}{\mathcal{O}\xspace}

\newcommand{\bbE}{\mathbb{E}\xspace}

\newcommand{\bbG}{\mathbb{G}\xspace}
\newcommand{\bbU}{\mathbb{U}\xspace}
\newcommand{\bbD}{\mathbb{D}\xspace}

\newcommand{\subvfg}{{\rm vfg}}

\newcommand{\vfreach}{\hookleftarrow}
\newcommand{\locvar}{{lv}\xspace}
\newcommand{\locvars}{\mathbb{V}\xspace}
\newcommand{\lvar}[1]{\langle#1\rangle\xspace}
\newcommand{\singlevf}{e\xspace}

\newcommand{\cupeq}{\ \cup\!\!=\xspace}

\newcommand{\ruledef}[3]{ #1 & $\frac{\begin{array}[c]{c}#2\end{array}}{\begin{array}[c]{c}#3\end{array}}$}

\newcommand{\rulename}[1]{\rulelab{\MakeUppercase{#1}}}
\newcommand{\rulelab}[1]{{\small\texttt{[{#1}]}}}

\newcommand{\mathmu}{\textcolor{green}{\boldsymbol{\mu}}}
\newcommand{\mathchi}{\textcolor{green}{\boldsymbol{\chi}}}
\newcommand{\muop}[1]{\textcolor{green}{\mathmu(#1)}}
\newcommand{\chiop}[2]{\textcolor{green}{#1\!=\!\mathchi(#2)}}

\SetAlFnt{\small}
\SetAlCapFnt{\small}
\SetAlCapNameFnt{\small}
\SetAlCapHSkip{0pt}
\IncMargin{-\parindent}

\acmVolume{9}
\acmNumber{4}
\acmArticle{39}
\acmYear{2010}
\acmMonth{3}

\doi{0000001.0000001}

\issn{1234-56789}

\makeatletter
\def\runningfoot{\def\@runningfoot{}}
\def\firstfoot{\def\@firstfoot{}}
\makeatother

\begin{document}

\markboth{Yulei Sui and Jingling Xue}{Demand-Driven Flow-Sensitive Pointer Analysis}

\title{Demand-Driven Pointer Analysis with Strong Updates via Value-Flow Refinement}
\author{Yulei Sui
\affil{School of Computer Science and Engineering, UNSW Australia}
Jingling Xue
\affil{School of Computer Science and Engineering, UNSW Australia}
}

\begin{abstract}
We present a new demand-driven flow- and context-sensitive pointer analysis with strong updates for C programs, 
called \vfsu, that enables computing points-to information via value-flow refinement, in environments with small time
and memory budgets such as IDEs. 
We formulate \vfsu by solving a graph-reachability problem 
on an inter-procedural value-flow graph representing a program's
def-use chains, which are pre-computed 
efficiently but over-approximately.
To answer a client query (a request for a variable's points-to
set), \vfsu reasons about the flow of values along the pre-computed
def-use chains sparsely
(rather than across all program points), by performing only the work
necessary for the query (rather than analyzing the whole program).
In particular, strong updates are performed to filter out
spurious def-use chains through value-flow refinement
as long as the total budget is not exhausted. 
\vfsu facilitates efficiency and precision tradeoffs by applying
different pointer analyses in a hybrid multi-stage analysis framework.

We have implemented \vfsu in LLVM (3.5.0) and evaluate it
by choosing uninitialized pointer detection as a major
client on 18 open-source C programs.
As the analysis budget increases,
\vfsu achieves improved precision, with its single-stage
flow-sensitive analysis reaching 97.4\% of that
achieved by whole-program flow-sensitive
analysis by consuming about 0.18 seconds
and 65KB of memory per query, on average (with a
budget of at most 10000 value-flow edges per query).
With context-sensitivity also considered, \vfsu's
two-stage analysis becomes more precise for some
programs but also incurs more analysis times. 
\vfsu is also amenable to parallelization.
A parallel implementation of its single-stage flow-sensitive analysis achieves a speedup of 
up to 6.9x with an average  of 3.05x 
a 8-core machine with respect
its sequential version.
\end{abstract}

%
%

%
%



\maketitle


\section{INTRODUCTION}

Pointer analysis is one of the most fundamental static program analyses, on which virtually
all others are built. The goal of pointer analysis is 
to compute an approximation of the set of abstract 
objects that 
a pointer can refer to.
A pointer analysis is (1) \emph{flow-sensitive} if it respects control flow
and \emph{flow-insensitive} otherwise and (2)
\emph{context-sensitive} if it distinguishes
different calling contexts 
and \emph{context-insensitive} otherwise.

\emph{Strong updates}, where stores overwrite, i.e., 
kill the previous contents of their abstract destination 
objects with new values, is an important factor
in the precision of pointer 
analysis~\cite{hardekopf2009semi,strongupdate}. 
In the case of \emph{weak updates}, these objects 
are assumed
conservatively to also retain their old contents.
Strong updates are possible only if flow-sensitivity
is maintained. In addition,
a flow-sensitive analysis can strongly update
an abstract object written at a store if and only if 
that object has 
exactly one concrete memory address, known as a 
singleton. 
By applying strong updates where needed, a pointer
analysis can improve
precision,
thereby providing significant benefits to many clients, such as 
change impact analysis~\cite{AchRob11},   
bug detection~\cite{dye14,Yan:2016:AML:2851613.2851773},
security analysis~\cite{Arzt:2014},
type state verification~\cite{fink2008effective}, 
compiler optimization \cite{sui2013query,sui2016loop,sui2014making}, and
symbolic execution~\cite{thresher-pldi13}.

In this paper, we introduce a demand-driven pointer
analysis for C by investigating how to perform 
strong updates effectively in a flow- and context-sensitive framework. For C programs, flow-sensitivity is 
important in achieving the precision required by the
afore-mentioned client applications due to strong
updates performed. If context-sensitivity is also
considered, some more strong updates are possible
for some programs at the expense of more analysis
times.
For object-oriented
languages like Java, context-sensitivity (without
strong updates) is widely used in achieving useful 
precision~\cite{spark,Li14,Ana02,Ana05,smaragdakispick,Sun:2011,xiao2011geometric}.

Ideally, strong updates at stores should be performed by analyzing all
paths independently by solving
a \emph{meet-over-all-paths}
(MOP) problem. 
However, even with branch conditions being
ignored, this problem
is intractable due to potentially unbounded number of 
paths that must be analyzed~\cite{landi1992undecidability,ramalingam1994undecidability}.

Instead, 
traditional flow-sensitive pointer analysis (\FS) for C
\cite{hind1998assessing,kam1977monotone} 
computes the
maximal-fixed-point solution (MFP)
as an over-approximation of MOP by solving an iterative
data-flow problem.  Thus,
the data-flow facts that reach a confluence point along
different paths are merged.  Improving on this,
sparse flow-sensitive pointer analysis (\SFS)~\cite{li2011boosting,hardekopfflow,oh2012design,ye2014region,yu2010level} 
boosts the performance of FS in
analyzing large C programs while maintaining the same strong updates
done by FS. The basic idea is to first conduct a pre-analysis on the
program to over-approximate its def-use chains and then perform FS by
propagating the data-flow facts, i.e., points-to information sparsely
along only the pre-computed def-use chains (aka value-flows) instead
of all program points in the program's control-flow
graph (CFG).

Recently, an 
approach~\cite{strongupdate} for performing
strong updates in C programs is introduced. It
sacrifices the precision of \FS to
gain efficiency by applying strong updates at stores where 
flow-sensitive singleton points-to sets are available but falls back 
to the flow-insensitive points-to information otherwise.
 
By nature, the challenge of pointer analysis is to make judicious
tradeoffs between efficiency and precision. 
Virtually all of the prior analyses for C
that consider some degree of flow-sensitivity  are 
whole-program analyses. Precise ones
are unscalable since they must typically
consider both flow- and context-sensitivity (\FSCS) in order to
maximize the number of strong updates performed. In contrast,
faster ones like \cite{strongupdate} are less precise,
due to both missing strong updates and propagating the points-to
information flow-insensitively across the weakly-updated locations.

In practice, a client application of a pointer
analysis may require only parts of the 
program to be analyzed. In addition, some points-to
queries may
demand precise answers while others can be answered 
as precisely as possible with
small time and memory budgets.
In all these cases,
performing strong updates blindly across the entire program
is cost-ineffective in achieving precision.

\begin{figure}[t]
\centering
\begin{tabular}{l}
\hfil\Large\textbf{\textcolor{blue}{}}\hfil\includegraphics[scale=0.75]{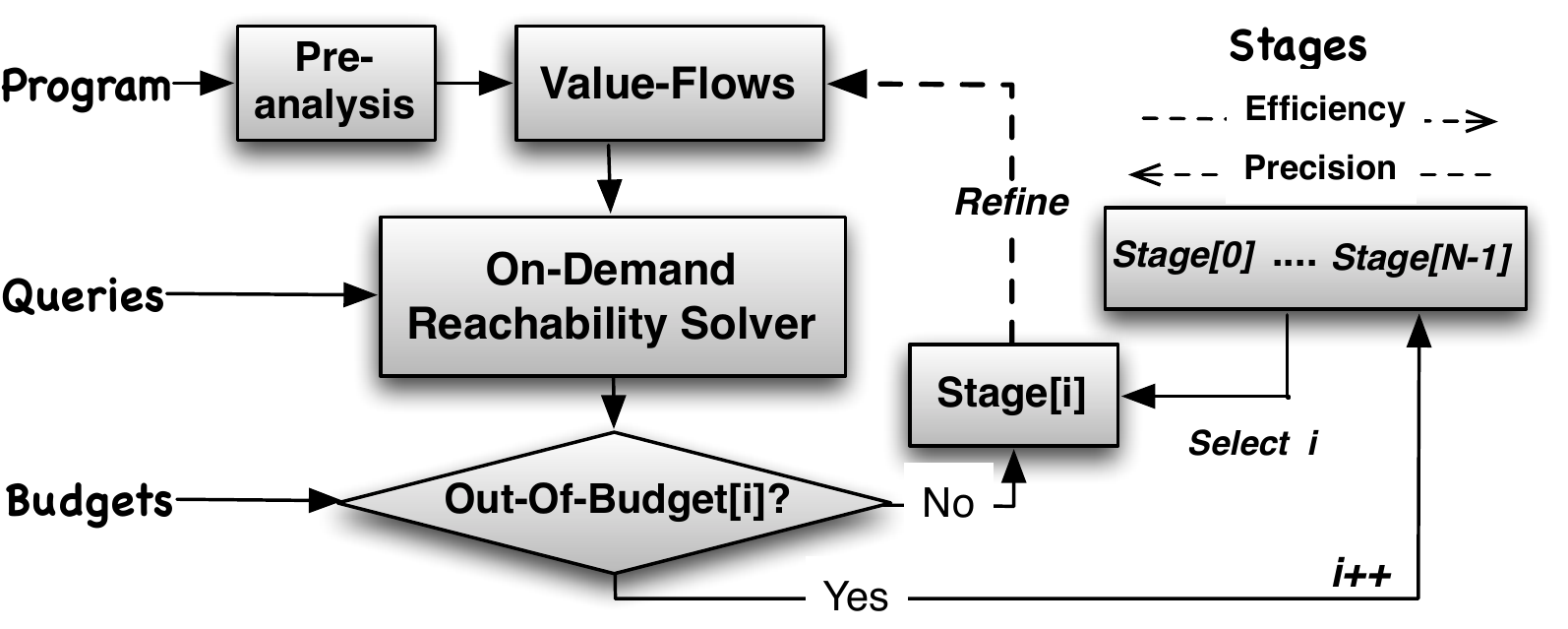}
\end{tabular}
\noindent\rule{8.5cm}{0.4pt}
\caption{Overview of \vfsu 
\label{fig:framework}}
\vspace*{3ex}
\end{figure}

For C programs,
how do we develop precise and efficient pointer analyses 
that are focused and partial, paying closer attention to the parts of 
the programs relevant to on-demand queries?
Demand-driven analyses for C
\cite{Heintze:2001:DPA,Zhang:2014:ESA,Zheng:2008} and  Java 
\cite{Lu13,Shang12,Sridharan:2006,Su15,yan2011demand} are flow-insensitive and thus cannot
perform strong updates to produce the
precision needed by some clients. 
\textsc{Boomerang}~\cite{spath2016boomerang} represents
a recent flow- and
context-sensitive demand-driven pointer analysis for Java.
However, its access-path-based approach 
performs strong updates at a store $a.f=\dots$ only
partially, by updating
$a.f$ strongly and the aliases of $a.f.*$ weakly.
Elsewhere, advances in whole-program flow-sensitive 
analysis for C have exploited some form of sparsity to improve
performance~\cite{hardekopfflow,li2011boosting,oh2012design,ye2014region,yu2010level}.
However, how to replicate this success 
for demand-driven flow-sensitive analysis for  C
is unclear. 
Finally, it remains open as to whether sparse 
strong update analysis can be performed both flow- and
context-sensitively
on-demand to avoid under- or over-analyzing.

In this paper, we introduce \vfsu, the first 
demand-driven pointer analysis with strong updates for
C, designed
to support flexible yet effective tradeoffs between 
efficiency and precision in answering client queries,
in environments with small time and memory budgets such
as IDEs.
As shown in Figure~\ref{fig:framework}, the
novelty behind \supa lies in
performing \textbf{S}trong \textbf{UP}date \textbf{A}nalysis precisely by refining imprecisely
pre-computed value-flows away in a hybrid multi-stage analysis framework.
Given a points-to query, strong updates are performed by 
solving a graph-reachability problem on an inter-procedural 
value-flow graph that captures the def-use chains of the program obtained
conservatively by a pre-analysis. Such over-approximated 
value-flows can be obtained by applying Andersen's analysis
\cite{andersen1994program} (flow- and context-insensitively). \vfsu conducts its
reachability analysis on-demand sparsely along only the pre-computed
value-flows rather than control-flows. In addition, \vfsu filters out 
imprecise value-flows by performing strong updates flow- and context-sensitively where needed
with no loss of precision as long as the total
analysis budget is sufficient.
The precision of \vfsu depends on the degree of
value-flow refinement performed under a budget. The more spurious 
value-flows \vfsu removes, the more precise 
the points-to facts are. 

\vfsu handles large C programs by staging analyses in
increasing efficiency but decreasing precision in a hybrid manner.
Currently, \vfsu proceeds in 
two stages by applying \FSCS and \FS in that order
with a configurable budget for each analysis. 
When failing to answer a query in a stage within its alloted
budget, \vfsu downgrades itself
to a more scalable but less precise analysis in the next stage
and will eventually fall back to the pre-computed
flow-insensitive information.
At each stage, \vfsu will re-answer the query by reusing
the points-to information 
found from processing the current and earlier
queries. By increasing the budgets used in the earlier
stages (e.g., \FSCS), \vfsu will obtain improved precision 
via improved value-flow refinement.

In summary, this paper makes the following contributions:
\begin{itemize}
\item We present the first demand-driven flow- and context-sensitive pointer analysis with strong updates  for C that enables computing precise points-to information by refining away imprecisely precomputed value-flows, subject to analysis budgets.

\item We introduce a hybrid multi-stage analysis framework that 
facilitates efficiency and precision tradeoffs
by staging different analyses in answering client queries.

\item We have produced an implementation
of \vfsu in LLVM (3.5.0) \cite{SUPA}. We evaluate \vfsu with
uninitialized pointer detection as a practical
client by using a total of 18 open-source C programs.
As the analysis budget increases,
\vfsu achieves improved precision, with its
single-stage flow-sensitive analysis reaching 97.4\%
of that achieved by whole-program flow-sensitive
analysis, by consuming about 0.18 seconds
and 65KB of memory per query, on average (with 
a per-query budget of at most 10000 value-flow edges
traversed).
With context-sensitivity also being considered, 
more strong updates are also possible. \vfsu's
two-stage analysis then becomes more precise for some
programs at the expense of more analysis times. 
\item We present four case studies to 
demonstrate that \vfsu is effective
in checking whether variables are initialized or not
in real-world applications.
\item We show that \vfsu is amenable to parallelization.
To demonstrate this, we have developed a parallel
implementation of \vfsu's single-stage flow-sensitive
analysis based on Intel Threading Building Blocks 
(TBB),  achieving a speedup of 
up to 6.9x with an average of 3.05x a 8-core machine over
its sequential version.

\end{itemize}

The rest of this paper is organized as follows.
Section~\ref{sec:bk} provides the background information.
Section~\ref{sec:mot} presents a motivating example. 
Section~\ref{sec:intra} introduces our formalism
for \vfsu.
Section~\ref{sec:eval} discusses and analyzes our
experimental results.
Section~\ref{sec:cases} contains four case studies.
Section~\ref{sec:para} describes a parallel
implementation of \vfsu.
Section~\ref{sec:rela} describes the related work.
Finally,
Section~\ref{sec:conc} concludes the paper.


\section{Background}
\label{sec:bk}

We describe how to represent a C program by an 
interprocedural
sparse
value-flow graph to enable demand-driven pointer analysis
via value-flow refinement. 
Section~\ref{sec:llvm-ir}
introduces the part of LLVM-IR relevant to pointer analysis.
Section~\ref{sec:top-ssa} describes how to put top-level
variables in SSA form.
Section~\ref{sec:addr-ssa} describes how to put address-taken
variables in SSA form.
Section~\ref{sec:svfg} constructs a
sparse value-flow graph that represents the
def-use chains for both top-level and address-taken variables
in the program.

\subsection{LLVM-IR}
\label{sec:llvm-ir}

We perform pointer analysis in the LLVM-IR of a program, as 
in~\cite{hardekopfflow,strongupdate,li2011boosting,Sui2012,ye2014region,George2016}.  
The domains and the LLVM instructions relevant to pointer
analysis are given in Table~\ref{tab:domain}.
The set of all variables $\calV$ are separated into 
two subsets, $\calO$ that contains all possible abstract objects, i.e., 
\emph{address-taken variables} of a pointer and 
$\calP$ that contains all \emph{top-level variables}. 

In LLVM-IR, top-level variables in $\calP = \calS \cup \calG$, including stack virtual registers (symbols starting with \textsf{"\%"}) and global variables (symbols starting with \textsf{"@"})
are explicit, i.e., directly accessed.
Address-taken variables in $\calO$ are implicit, i.e.,
accessed indirectly 
at LLVM's \textsf{load} or \textsf{store} instructions
via top-level variables.

Only a subset of the complete LLVM instruction set that is relevant to pointer analysis are modeled. As in Table~\ref{tab:domain}, every function $f$ of a program contains nine types of instructions (statements), including seven types of instructions used in the function body of $f$, and one \textsc{FunEntry} instruction $f(r_1,\dots,r_n)$ with the declarations of the parameters of $f$, and one \textsc{FunExit} instruction $ret_f \ p$ as the unique return of $f$. Note that the
LLVM pass \textsf{UnifyFunctionExitNodes} is executed before pointer analysis in order 
to ensure that every function has only one \textsc{FunExit} instruction.

\begin{table}[h]
\centering
\caption{Domains and LLVM instructions used by 
pointer analysis.
\label{tab:domain}
}
\addtolength{\tabcolsep}{-.3ex}
\renewcommand{\arraystretch}{1.2}
\begin{tabular}{@{}l|l@{}}
	\hline
\multicolumn{1}{c|}{Analysis Domains} & \multicolumn{1}{c}{LLVM Instruction Set} \\
	\hline
	{
\renewcommand{\arraystretch}{1.31}
\begin{tabular}[t]{l@{\hspace{1mm}}l@{\hspace{1mm}}ll}
$\lab$&$\in$ & $ \calL$ & instruction labels \\
$fld$&$\in$ & $ \calC$ & constants (field accesses) \\
$s$&$\in$ & $\calS$ & stack virtual registers \\
$g$&$\in$ & $ \calG$ & global variables \\
$f$&$\in$ & $ \calF \subseteq \calG$ & program functions \\
$p,q,r,x,y$&$\in$ & $ \calP = \calS \cup \calG$ & top-level variables\\
$o,a,b,c,d$&$\in$ & $ \calO$ &address-taken variables \\ 
$v$&$\in$ & $ \calV = \calP \cup \calO$ & program variables \\
\end{tabular}
}
&

\begin{tabular}[t]{@{}ll@{}l@{}l}
\textsc{AddrOf} & $p$ &\ =\ & $\ \&o$ \\
\textsc{Copy} & $p$ &\ =\ & $\ q$ \\
\textsc{Phi} & $p$ &\ =\ & $\phi(q,r)$ \\
\textsc{Field} & $p$ &\ =\ & $\ \&q\!\rightarrow\!fld$ \\
\textsc{Load} & $p$ &\ =\ & $\ *q$\\
\textsc{Store} &  $*p$ &\ =\ & $\ q$ \\
\textsc{Call} &$p$ &\ =\ & $\ q(r_1,\dots,r_n)$ \\
\textsc{FunEntry} &\multicolumn{ 3}{l}{$f(r_1,\dots,r_n)$}  \\
\textsc{FunExit} &\multicolumn{ 3}{l}{$ret_f \ p$}  \ \\
\end{tabular}\\
\hline
\end{tabular}
\end{table}

Let us go through the seven types of instructions used inside
a function.
For an \textsc{AddrOf} instruction $p\!\!=\!\!\& o$, 
known as an \emph{allocation site}, $o$ is 
one of the following objects: 
(1) a stack object, $o_\lab$, where $\lab$ is its allocation site (via an LLVM 
\textsf{alloca} instruction), (2) a global object, i.e.,
a global object $o_\lab$, where $\lab$ is its allocation
site or a program function $o_f$, where $f$ is its name, and
(3) a dynamically created heap object $o^h_\lab$,
where $\lab$ is its heap allocation site (e.g., via a
\textsf{malloc() call}). 
For each object $o$ (except for a function), we 
write $o_{fld}$ to represent the sub-object that 
corresponds to its field $fld$.
For flow-sensitive pointer analysis, the initializations 
for global objects take place at the entry of \textsf{main()}. 


\textsc{Copy} denotes a casting instruction (e.g., \textsf{bitcast}) in LLVM. \textsc{Phi} is a standard SSA instruction introduced at a confluence point in the CFG to select the value of a variable from different control-flow branches.
\textsc{Load} (\textsc{Store}) is a memory accessing 
instruction that reads (write) a value from (into)
an address-taken object.

Our handling of field-sensitivity is ANSI-compliant. Given a pointer to an aggregate (e.g., a struct or an array), 
pointer arithmetic used for accessing
anything other than the aggregate itself has undefined behavior~\cite{pearce2007efficient,iso90} and thus not modeled.
To model the field accesses of a struct object, \textsc{Field} represents a \textsf{getelementptr} instruction with its field offset $fld$ as a constant value.  A 
\textsf{getelementptr} instruction that operates on
a non-constant field of
a struct is modeled as \textsc{Copy} instructions, one
for every field of the struct conservatively. Arrays are treated monolithically.

\textsc{Call}, $p=q(r_1,\dots,r_n)$, denotes a \textsf{call} 
instruction, where $q$ can be either a global variable (for a direct call) or a stack virtual register (for
an indirect call). 

\subsection{SSA Form for Top-Level Variables}
\label{sec:top-ssa}

LLVM-IR is known as a partial SSA form since
only top-level variables are in SSA form. In LLVM-IR,
top-level variables are explicit, i.e., directly accessed 
and can thus be put in SSA form by using a standard SSA
construction algorithm~\cite{cytron1991efficiently}
(with \textsc{Phi} instructions inserted at confluence points).


\begin{figure}[th]
\centering
\begin{tabular}{c}
\hfil\Large\textbf{\textcolor{blue}{}}\hfil\includegraphics[scale=0.72]{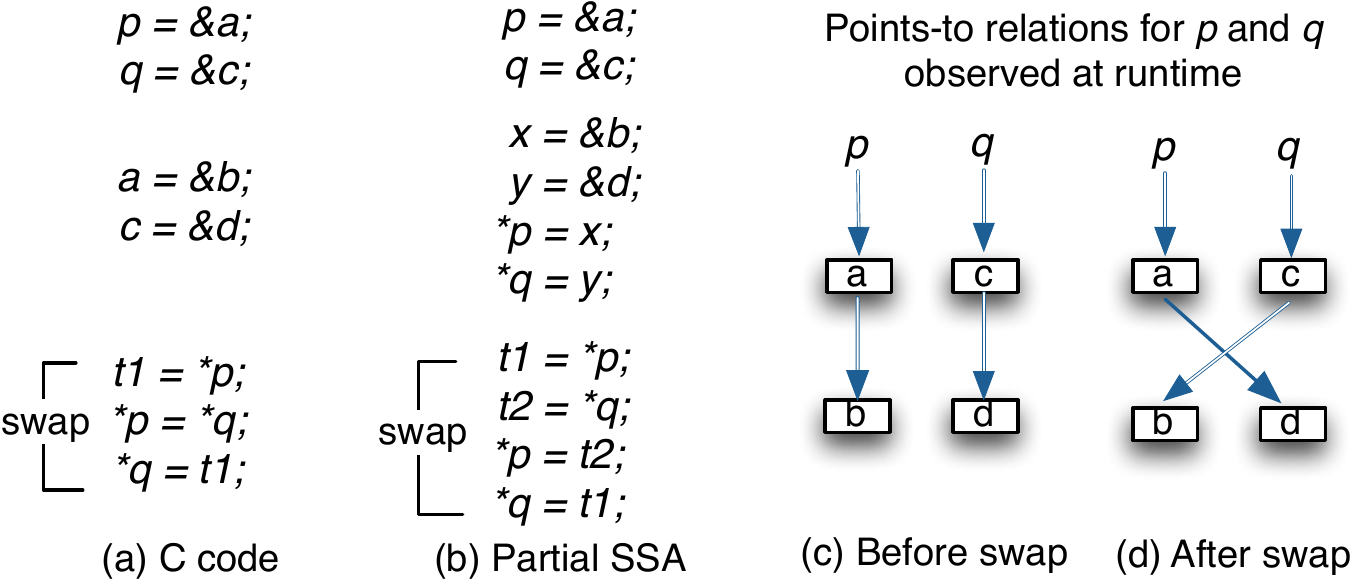}
\\ \hline
\end{tabular} 
\caption{A \textsf{swap} example and its partial SSA form.
	\label{fig:bkgex} }
\end{figure}

Let us illustrate LLVM's partial SSA form by using an example
in Figure~\ref{fig:bkgex}.
Figure~\ref{fig:bkgex}(a) shows a \textsf{swap} program in C
and Figure~\ref{fig:bkgex}(b) gives
its corresponding partial SSA form. 
Figures~\ref{fig:bkgex}(c) and (d) depict some (runtime)
points-to relations before and after the swap operation.
In this example, we have
$p,q,x,y,t1,t2 \in \calP$ and $a,b,c,d \in \calO$. Note that 
$x,y,t1$ and $t2$ are new temporary registers introduced in
order to put the program given
in Figure~\ref{fig:bkgex}(a) into the partial SSA form
given in Figure~\ref{fig:bkgex}(b). 
In particular,
$*p=*q$ is decomposed into 
$t2=*q$ and $*p=t2$, where 
$t2$ is a top-level pointer.

In LLVM-IR, all top-level variables are in SSA form.
In this example, all top-level variables are trivially in
SSA form, as each has exactly one definition only.
As a result, the def-use chains for top-level variables
are readily available.

However, address-taken variables are accessed
indirectly at loads and stores via top-level variables
and thus  not in SSA form. For example, the address-taken
variable
$a$ is defined implicitly twice, once at $*p = x$ and once
at $*p=t2$, and the address-taken variable
$c$ is also defined implicitly twice, once at $*q = y$ and once
at $*q=t1$. As a result, the def-use chains for 
address-taken variables are not immediately available.

\subsection{SSA Form for Address-Taken Variables}
\label{sec:addr-ssa}

Starting with LLVM's partial SSA form, we first perform a
pre-analysis by using 
Andersen's algorithm flow- and context-insensitively~\cite{andersen1994program}, 
implemented in SVF~\cite{SVF}. We then put address-taken
variables in \emph{memory SSA form}, by using the SSA
construction algorithm~\cite{cytron1991efficiently}. Imprecise points-to
information computed this way will be refined by
our demand-driven pointer analysis.

Given a variable $v$, $\apts(v)$ represents its points-to set
computed by Andersen's algorithm.
There are two steps ~\cite{sui2014detecting}, illustrated
in Figures~\ref{fig:bkg-intra}(a) and (b) intraprocedurally
and in Figures~\ref{fig:bkg-inter}(a) and (b) interprocedurally.

\begin{description}
\item [Step 1: Computing Modification and Reference Side-Effects] 
As shown in Figure~\ref{fig:bkg-intra}(a),
every load, e.g., $t1=*q$ 
is annotated with a $\muop{a}$ operator for each object $a$ 
pointed by $q$, i.e., $a\in\apts(q)$ to represent a potential use of $a$ at the load. 
Similarly,
every store, e.g., $*p=x$ is annotated with a 
$\chiop{a}{a}$ operator for each object $a\in\apts(p)$ 
to represent a potential def and use of $a$ at the store. If $a$ can be \emph{strongly updated}, then $a$ receives whatever $x$ points to and the old contents in $a$ are killed. Otherwise, $a$ must also incorporate its old contents, resulting in a 
\emph{weak update} to $a$. 

\begin{figure}[ht]
\centering
\vspace*{-1ex}
\begin{tabular}{c}
\hfil\Large\textbf{\textcolor{blue}{}}\hfil\includegraphics[scale=0.70]{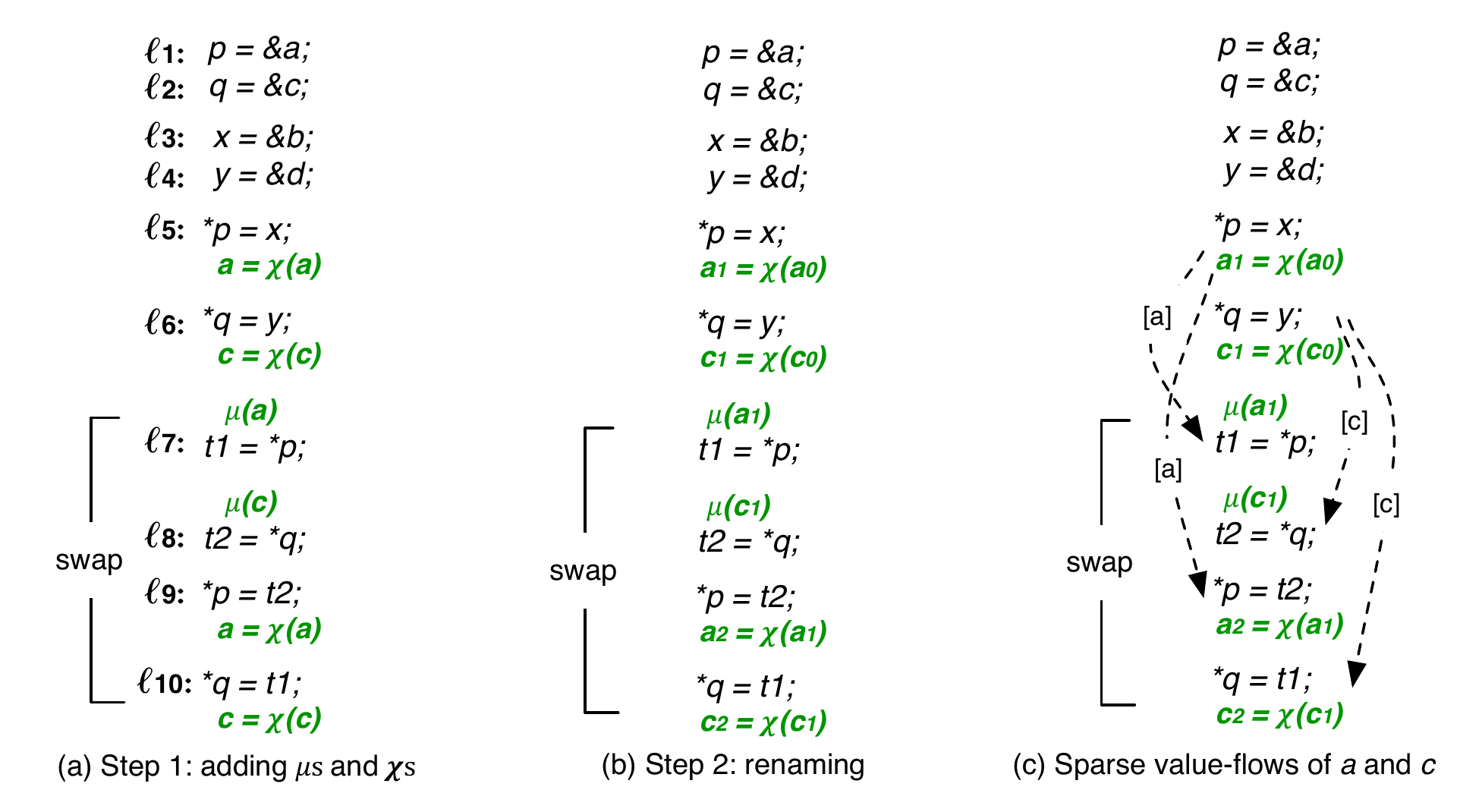}
\\[-2ex] \hline
\end{tabular} 
\caption{Memory SSA form 
and sparse value-flows constructed intraprocedurally for
Figure~\ref{fig:bkgex},
obtained with Andersen's analysis: $\apts(p)= \{a\}$ and $\apts(q) = \{c\}$.
	\label{fig:bkg-intra} }
\end{figure}

\begin{figure}[h]
\hspace*{0ex}
\begin{tabular}{@{\hspace{-3mm}}l@{\hspace{-5mm}}}
\hfil\Large\textbf{\textcolor{blue}{}}\hfil\includegraphics[scale=0.67]{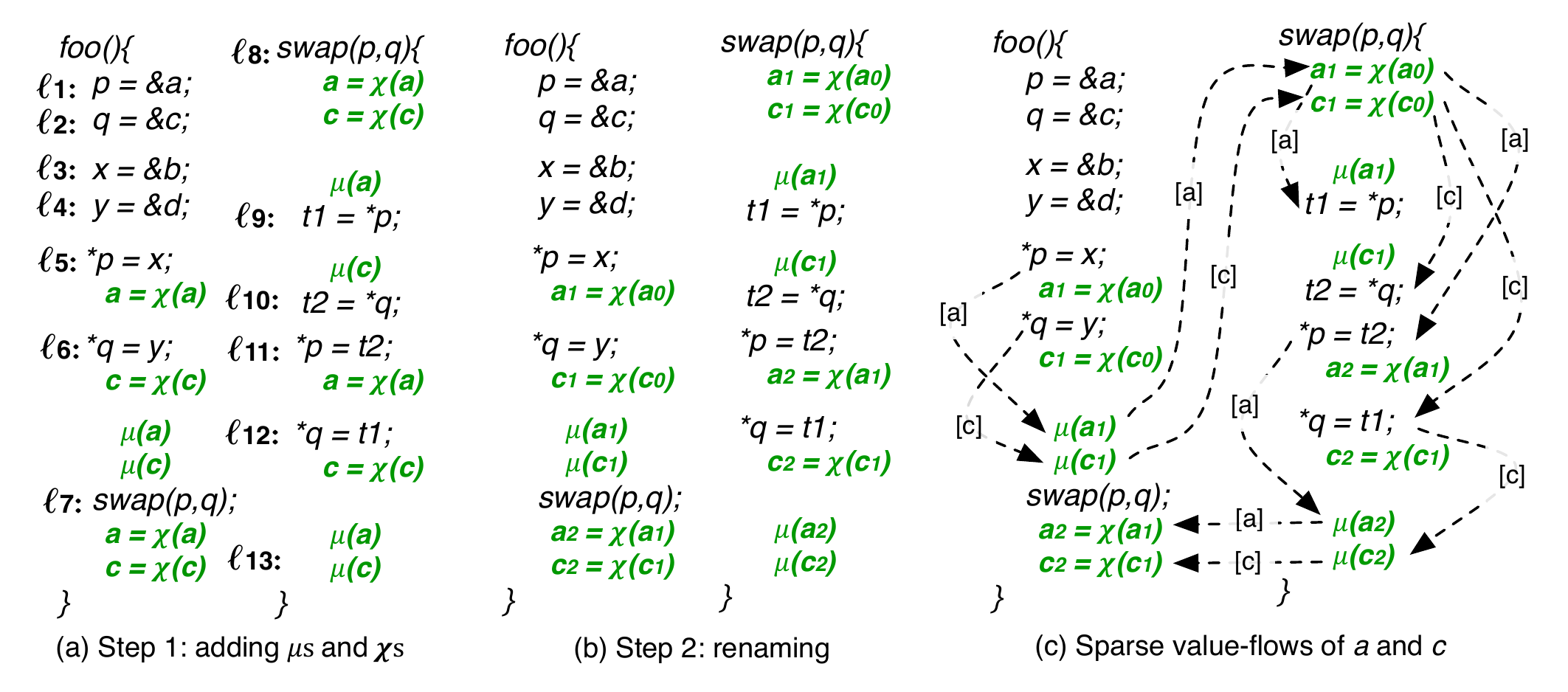}
\\[-2ex] \hline
\end{tabular} 
\caption{Memory SSA form and sparse value-flows constructed
	interprocedurally for an example 
modified from Figure~\ref{fig:bkgex} with its
four swap instructions moved into a separate function,
called \textsf{swap}. $\lab_8$ and $\lab_{13}$ correspond to the \textsc{FunEntry} and \textsc{FunExit} of \textsf{swap}.
	\label{fig:bkg-inter} }
	\vspace*{-3mm}
\end{figure}

We compute the side-effects of a function call by applying
a lightweight interprocedural mod-ref analysis~\cite
[\textcolor{blue}{\S 4.2.1}]{sui2014detecting}.
For a given callsite $\lab$, it is annotated with 
$\muop{a}$ ($\chiop{a}{a}$) 
if $a$ may be read (modified) inside the callees of $\lab$ (discovered by Andersen's pointer analysis). In addition, 
appropriate $\mathchi$ and $\mathmu$ operators are also added 
for the
\textsc{FunEntry} and \textsc{FunExit} instructions of these
callees in order to mimic passing parameters 
and returning results for address-taken variables. 

Figure~\ref{fig:bkg-inter}(a) gives an example modified
from Figure~\ref{fig:bkg-intra}(a) by moving
the four swap instructions into a function, 
\textsf{swap}. For read side-effects, $\muop{a}$ and $\muop{c}$ are added 
before callsite $\lab_7$ to represent the potential uses of $a$ and $c$ in \textsf{swap}. Correspondingly, \textsf{swap}'s \textsc{FunEntry} instruction $\lab_8$ is annotated with $\chiop{a}{a}$ and $\chiop{c}{c}$ to receive the values of $a$ and $c$ 
passed from $\lab_7$. For modification side-effects, $\chiop{a}{a}$ and $\chiop{c}{c}$ are added after $\lab_7$ to 
receive the potentially modified values of $a$ and $c$ returned
from \textsf{swap}'s \textsc{FunExit} instruction $\lab_{13}$,
which are annotated with $\muop{a}$ and $\muop{c}$.


\item [Step 2: Memory SSA Renaming] 
All the address-taken variables are converted into SSA form
as suggested in~\cite{chow1996effective}. 
Every $\muop{a}$ is treated as a use of $a$.
Every $\chiop{a}{a}$ is treated as both a def and use 
of $a$, as $a$ may admit only a weak update. 
Then the SSA form for address-taken variables is obtained
by applying a standard SSA 
construction algorithm~\cite{cytron1991efficiently}.

For the program annotated with $\mu$'s and $\chi$'s in
Figure~\ref{fig:bkg-intra}(a), 
Figure~\ref{fig:bkg-intra}(b) gives its memory SSA form.
Similarly, Figure~\ref{fig:bkg-inter}(b) gives the
memory SSA form for 
Figure~\ref{fig:bkg-inter}(a).

\end{description}

\begin{table}[t]
	\centering
\caption{Def-use information of both top-level and address-taken variables. $Def_{v}$ ($Use_{v}$) denotes the set of definition (use) instructions for a variable $v\in\calV$. 
\label{tab:def-use}}
\begin{tabular}{@{}l@{\hspace{0mm}}|| l}
\hline
\multicolumn{1}{c}{Instruction} \ $\lab$ & \multicolumn{1}{c}{Defs and Uses of Variables in Memory 
SSA Form}\\ \hline
$p = \&o$   & $\{\lab\} =  Def_{p}$ \\
$p = q$   & $\{\lab\} = Def_{p}$ \quad $\lab\in Use_{q}$ \\
$p = \phi(q,r)$   & $\{\lab\} =   Def_{p}$ \quad $\lab\in Use_{q}$ \quad $\lab\in Use_{r}$ \\
$p =\&q\!\rightarrow\!fld$   & $\{\lab\} =   Def_{p}$ \quad $\lab\in Use_{q}$ \\
 \hline
$p = *q$ \hspace{9.5mm}\quad $\muop{a_i}$ & $\{\lab\} =   Def_{p}$ \quad $\lab\in Use_{q}$ \quad $\lab\in Use_{a_i}$ \\
$*p = q$  \hspace{8.5mm} \quad $\chiop{a_{i+1}}{a_i}$  & $\lab\in Use_{p}$ \quad $\lab\in Use_{q}$ \quad $\lab\in Def_{a_{i+1}}$ \quad $\lab\in Use_{a_i}$ \\
\hline
$p = q(r_1,\dots,r_n)$   & $\{\lab\} =   Def_{p}$ \quad $\lab\in Use_{q}$ \quad $\forall\ i\in1,\dots,n:\!\ \!\lab\in Use_{r_i}$\\
 \quad \ $\muop{a_i}$ \quad $\chiop{a_{j+1}}{a_j}$ & $\lab\in Use_{a_i}$ \quad $\lab\in Def_{a_{j+1}}$ \quad $\lab\in Use_{a_j}$ \\
 \hline
 $f(r_1,\dots,r_n) $  \hspace{2mm} $\chiop{a_{i+1}}{a_i}$ & $\forall\ i\in1,\dots,n:\!\ \!\lab\in Def_{r_i}$ \ \ $\lab\in Def_{a_{i+1}}$ \ \ $\lab\in Use_{a_i}$ \\
 $ret_f \ p$ \hspace{8.5mm} \quad $\muop{a_i}$ &  $\lab\in Use_{p}$ \quad $\lab\in Use_{a_i}$ \\ \hline
\end{tabular}
\end{table}

\begin{figure}[h]
	\centering
\hspace*{-2.5ex}
\begin{tabular}{l}
\begin{tabular}{ll}
\\
\ruledef{\rulename{Intra-Top}\hspace{-4mm}}{
\begin{tabular}{lll}
$\lab \in Def_p$ \quad $\lab' \in Use_p$   \\
\end{tabular}
}{
\begin{tabular}{ll}
$\lab\vfedge{p}\lab'$\\
\end{tabular}
}
\end{tabular}

\quad

\begin{tabular}{ll}
\\
\ruledef{\rulename{Intra-Addr}\hspace{-4mm}}{
\begin{tabular}{lll}
$\lab \in Def_{a_i}$ \quad $\lab' \in Use_{a_i}$   \\
\end{tabular}
}{
\begin{tabular}{ll}
$\lab\vfedge{a}\lab'$\\
\end{tabular}
}
\end{tabular}

\\\\[-2ex]

\begin{tabular}{ll}
\ruledef{\rulename{Inter-call-Top}\hspace{-4mm}}{
\begin{tabular}{lll}
$\lab: p = q(r_1,\dots,r_n)$ \quad $o_f\in\apts(q)$ \quad $\lab': f(r'_1,\dots,r'_n)$   \\
\end{tabular}
}{
\begin{tabular}{ll}
$\forall i\in 1,\dots,n: \lab\vfedge{r_i}\lab'$\\
\end{tabular}
}
\end{tabular}

\\ \\[-2ex]

\begin{tabular}{ll}
\ruledef{\rulename{Inter-ret-Top}\hspace{-4mm}}{
\begin{tabular}{lll}
$\lab: p = q(\dots)$  \quad $a_f\in\apts(q)$ \quad $\lab': ret_f\ p'$   \\
\end{tabular}
}{
\begin{tabular}{ll}
$\lab'\vfedge{p}\lab$\\
\end{tabular}
}
\end{tabular}

\\\\[-2ex]

\begin{tabular}{ll}
\ruledef{\rulename{Inter-call-Addr}\hspace{-4mm}}{
\begin{tabular}{lll}
$\lab: p = q(\dots)$ $\textcolor{blue}{\muop{a_i}}$ \quad $a_f\in\apts(q)$ \quad $\lab': f(\dots)$ $\textcolor{blue}{\chiop{a_{j+1}}{a_j}}$   \\
\end{tabular}
}{
\begin{tabular}{ll}
$\lab\vfedge{a}\lab'$\\
\end{tabular}
}
\end{tabular}

\\ \\[-2ex]

\begin{tabular}{ll}
\ruledef{\rulename{Inter-ret-Addr}\hspace{-4mm}}{
\begin{tabular}{lll}
$\lab: \_ = q(\dots)$ $\chiop{a_{j+1}}{a_j}$  \quad $a_f\in\apts(q)$ \quad  $\lab': ret_f \ \_ \quad \textcolor{blue}{\muop{a_i}}$  \\
\end{tabular}
}{
\begin{tabular}{ll}
$\lab'\vfedge{a}\lab$\\
\end{tabular}
}
\end{tabular}
\end{tabular}
\caption{Value-flow construction in Memory SSA form. \label{fig:pre-vf-rules}
}
\end{figure}

\subsection{Sparse Value-Flow Graph} 
\label{sec:svfg}

Once both top-level and address-taken variables are in SSA
form, their def-use chains are immediately available,
as shown in Table~\ref{tab:def-use}.
We discussed top-level variables earlier.
For the two address-taken variables $a$ and $c$ 
in Figure~\ref{fig:bkgex},
Figure~\ref{fig:bkg-intra}(c) depicts their def-use chains,
i.e., sparse value-flows for the memory SSA form in
Figure~\ref{fig:bkg-intra}(b).
Similarly, Figure~\ref{fig:bkg-inter}(c) gives their
sparse value-flows for the
memory SSA form in
Figure~\ref{fig:bkg-inter}(b).

Given a program, a \emph{sparse value-flow graph} (SVFG),
$G_{\subvfg} = (N,E)$, is a multi-edged directed graph that 
captures its def-use chains for both top-level and
address-taken variables. $N$ is the set of nodes representing all instructions
and $E$ is the set of edges representing all potential def-use chains.
In particular, an edge $\lab_1 \vfedge{v}\lab_2$, where
$v\in\calV$, from statement $\lab_1$ to statement $\lab_2$
signifies a potential def-use chain for $v$ with its def
at $\lab_1$
and use at $\lab_2$. We refer to 
$\lab_1 \vfedge{v}\lab_2$ a \emph{direct value-flow} 
if $v\in\calP$ and an \emph{indirect value-flow} if $v\in\calO$.
This representation is \emph{sparse} since
the intermediate program points 
between $\ell_1$ and $\ell_2$ are omitted, thereby enabling
the underlying points-to information to be gradually refined by
applying
a sparse demand-driven pointer analysis.

Figure~\ref{fig:pre-vf-rules} gives the rules for connecting 
value-flows between two instructions based on the defs and uses
computed in Table~\ref{tab:def-use}. For
intraprocedural value-flows, 
\rulename{Intra-Top} and \rulename{Intra-Addr} handle
top-level and address-taken variables, respectively. In
SSA form, every use of a variable only has a unique definition.
For a use of $a$ identified as $a_i$ (with its $i$-th version) at $\lab'$ annotated with $\muop{a_i}$, its unique definition in SSA form is $a_i$ at an $\lab$ annotated with
$\chiop{a_{i}}{a_{i-1}}$.
Then, $\lab \vfedge{a}\lab'$ is generated to represent potentially the value-flow of $a$ from $\lab$ to $\lab'$. 
Thus,
the PHI functions introduced for address-taken variables will 
be ignored, as the value $a$ in $\lab \vfedge{a}\lab'$ is not versioned.


Let us consider interprocedural value-flows.
The def-use information in Table~\ref{tab:def-use} is only 
intraprocedural. According to Figure~\ref{fig:pre-vf-rules},
interprocedural value-flows
are constructed to represent parameter passing for
top-level variables (\rulename{Inter-call-Top} and
\rulename{Inter-ret-Top}), and 
the $\mathmu/\mathchi$ operators annotated 
at \textsc{FunEntry}, \textsc{FunExit} and \textsc{Call} 
for address-taken variables
(\rulename{Inter-call-addr} and \rulename{Inter-ret-addr}).

\rulename{Inter-call-Top} connects the value-flow 
from an actual argument $r_i$ at a
call instruction $\lab$ to its corresponding formal parameter
$r'_i$ at the \textsc{FunEntry} $\lab'$ of every callee $f$ invoked
at the call. Conversely, \rulename{Inter-ret-Top} models the 
value-flow from the \textsc{FunExit} instruction of $f$ to
every callsite where $f$ is invoked.
Just like for top-level variables, \rulename{Inter-call-addr} and \rulename{Inter-ret-addr} build the value-flows of 
address-taken variables across the functions according to the annotated $\mathmu$'s and $\mathchi$'s. 
Note that the versions $i$ and $j$ of an SSA variable $a$ 
in different functions may be different. For example, 
Figure~\ref{fig:bkg-inter}(c) illustrates the four inter-procedural value-flows $\lab_7 \vfedge{a} \lab_8$, $\lab_7 \vfedge{c} \lab_8$, $\lab_{13} \vfedge{a} \lab_7$ and $\lab_{13} \vfedge{c} \lab_7$ obtained by applying the two rules to
Figure~\ref{fig:bkg-inter}(b).

The SVFG obtained this way may contain spurious 
def-use chains, such as 
$\lab_5\vfedge{a}\lab_9$ in Figure~\ref{fig:bkg-intra},
as Andersen's flow- and context-insensitive pointer
analysis is fast but imprecise.
However, this representation allows imprecise points-to
information to be refined by performing sparse 
whole-program flow-sensitive pointer analysis
as in prior work~\cite{hardekopfflow,Nagaraj:2013,ye2014region,sui2016sparse}. In this paper,
we introduce a demand-driven flow- and context-sensitive
pointer analysis with strong updates that can
answer points-to queries efficiently and precisely on-demand, by
removing spurious def-use chains in the SVFG iteratively.


\begin{figure*}[th]
\centering
\hspace*{-3ex}
\begin{tabular}{l@{\hspace{10mm}}l}
\begin{tabular}{l}
\hfil\Large\textbf{\textcolor{blue}{}}\hfil\includegraphics[scale=0.4]{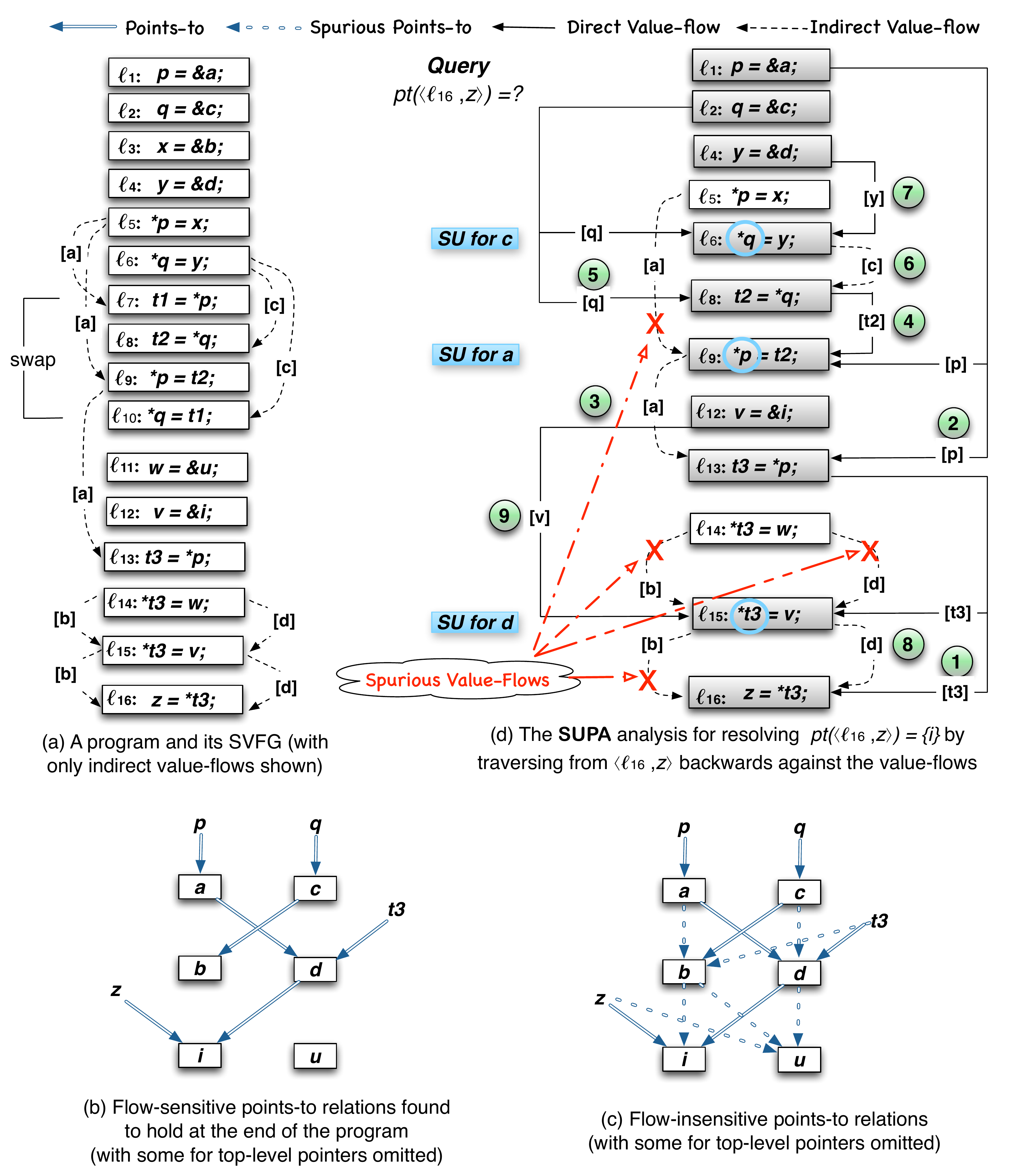}
\end{tabular}
\end{tabular}
\caption{A motivating example for illustrating \vfsu (SU 
	stands for ``Strong Update'').
\label{fig:motex}}
\end{figure*}

\section{A MOTIVATING EXAMPLE}
\label{sec:mot}

Our demand-driven pointer analysis, \vfsu, operates on the
SVFG of a program. It computes points-to
queries flow- and context-sensitively on-demand by performing
strong updates, whenever possible,
to refine away imprecise value-flows in the SVFG.  

Our example program, shown in Figure~\ref{fig:motex}(a),
is simple (even with 16 lines).  The program
consists of a straight-line sequence of code, with $\ell_1$ --
$\ell_{10}$ taken directly from Figure~\ref{fig:bkgex}(b) and the
six new statements
$\ell_{11}$ -- $\ell_{16}$ added to enable us to highlight
some key properties of \vfsu. 
We assume that $u$ at $\lab_{11}$ is uninitialized but $i$ at $\lab_{12}$ 
is initialized. The SVFG embedded 
in Figure~\ref{fig:motex}(a) will be referred to 
shortly below. We describe how \vfsu can be used to 
prove that
$z$  at $\lab_{16}$ points only to the initialized 
object $i$, by computing flow-sensitively on-demand the points-to query
$\pts(\lvar{\lab_{16},z})$, i.e.,
the points-to set
of $z$ at the program point after $\ell_{16}$, which is
defined in (\ref{eq:ptr}) in Section~\ref{sec:intra}. 

Figure~\ref{fig:motex}(b) depicts the points-to relations for the six
address-taken variables and some top-level ones found
at the \emph{end} of the code sequence
by a whole-program flow-sensitive analysis (with strong updates) like \SFS \cite{hardekopfflow}.
Due to flow-sensitivity, multiple solutions for a 
pointer are maintained.
In this example, these are the true relations observed
at the end of program execution. Note that \SFS 
gives rise to 
Figure~\ref{fig:bkgex}(c) by analyzing
$\ell_1$ -- $\ell_6$, Figure~\ref{fig:bkgex}(d) by 
analyzing also $\ell_7$ -- $\ell_{10}$, and finally,
Figure~\ref{fig:motex}(b) by analyzing 
$\ell_{11}$ -- $\ell_{16}$ further.
As $z$ points to $i$ but not $u$, no warning is
issued for $z$, implying that $z$ is regarded as being
properly initialized.

Figure~\ref{fig:motex}(c) shows how the points-to
relations in Figure~\ref{fig:motex}(b) are
over-approximated flow-insensitively by applying Andersen's analysis
\cite{andersen1994program}.
In this case, a single solution is computed
conservatively for the entire program.
Due to the lack of strong updates in analyzing the two stores
performed
by \textsf{swap}, the points-to relations in
Figures~\ref{fig:bkgex}(c) and  \ref{fig:bkgex}(d) are 
merged, causing
$*a$ and $*c$ to become spurious aliases. 
When $\ell_{11}$ -- $\ell_{16}$ are analyzed, the seven 
spurious points-to relations (shown in dashed
arrows
in Figure~\ref{fig:motex}(c)) are introduced. 
Since $z$ points to $i$ (correctly) and $u$
(spuriously), a false alarm for $z$ will be issued.
Failing to consider flow-sensitivity, Andersen's analysis
is not precise for this uninitialization pointer
detection client.

Let us now explain how \vfsu, shown 
in Figure~\ref{fig:framework},
works. \vfsu will first perform a pre-analysis to the example program to build the SVFG given
in Figure~\ref{fig:motex}(a), as discussed in
Section~\ref{sec:bk}. For its top-level variables,
their direct value-flows, i.e., def-use chains 
are explicit and thus omitted to avoid cluttering. 
For example, $q$  has three def-use chains
$\lab_2 \vfedge{q}\lab_6$,
$\lab_2 \vfedge{q}\lab_8$ and
$\lab_2 \vfedge{q}\lab_{10}$.
For its address-taken variables, there are
nine indirect value-flows,
i.e., def-use chains depicted
in Figure~\ref{fig:motex}(a).
Let us see how the two def-use chains for $b$ are created.
As $t3$ points to $b$,
$\ell_{14}$,  $\ell_{15}$ and  $\ell_{16}$ will be annotated with
$b=\chi(b)$, $b=\chi(b)$ and $\mu(b)$, respectively. By putting $b$
in SSA form, these three functions become
$b2=\chi(b1)$, $b3=\chi(b2)$ and $\mu(b3)$. Hence, we have
$\lab_{14}\! \vfedge{b}\!\lab_{15}$ and
$\lab_{15}\! \vfedge{b}\!\lab_{16}$, indicating $b$ at $\ell_{16}$
has two potential definitions, with the 
one at $\ell_{15}$ overwriting the one at $\ell_{14}$. 
The def-use chains for $d$ and $a$ are built similarly.

Let us consider a single-stage analysis
with $\textbf{\textit{Stage[N-1]}}=
\textbf{\textit{Stage[0]}}\!=\!FS$
in Figure~\ref{fig:framework}.
Figure~\ref{fig:motex}(d) shows how
\vfsu computes $\pts(\lvar{\ell_{16},z})$ on-demand, starting from $\ell_{16}$,
by performing a backward reachability analysis on the 
SVFG, with
the visiting order of def-use chains marked as
\circled{1} -- \circled{9}. Formally, this is done 
as illustrated in Figure~\ref{fig:mottrace}.
The def-use chains for only the relevant
top-level variables are shown.
By filtering out the four spurious value-flows (marked by \Cross),
\vfsu finds that
only $i$ at $\lab_{12}$ is backward reachable from $z$ 
at $\lab_{16}$. Thus, $\pts(\lvar{\lab_{16},z}) = \{i\}$. 
So no warning for $z$ will be issued.

\vfsu differs from prior work in the following 
three major aspects:
\begin{itemize}
\item \textbf{\emph{On-Demand Strong Updates}}

A whole-program flow-sensitive analysis like \SFS~\cite{hardekopfflow}
can answer $\pts(\lvar{\lab_{16},z})$ precisely but 
must accomplish this task by analyzing all the 16 statements, 
resulting in a total of six strong updates
performed at the six stores, with some strong updates performed unnecessarily for this query. Unfortunately, existing whole-program 
FSCS or even just FS algorithms
do not scale well for large C programs~\cite{AchRob11}.

In contrast, \vfsu computes
$\pts(\lvar{\lab_{16},z})$ precisely by performing
only three
strong updates at $\ell_6$, $\ell_9$ and $\ell_{15}$. 
The earlier a strong update is performed by \vfsu during its
reachability analysis, the fewer the number of
statements traversed. 
After \circled{1} -- \circled{8} have been
performed, \vfsu finds that 
$t3$ points to $d$ only. With a 
strong update performed at $\lab_{15}: *t3 = v$ (\circled{9}),
\vfsu concludes that $\pts(\lvar{\ell_{16},z})\! =\! \{i\}$.

\item {\textbf{\emph{Value-Flow Refinement}}}

Demand-driven  pointer
analyses~\cite{Shang12,Sridharan:2006,yan2011demand,Zhang:2014:ESA,Zheng:2008} are
flow-insensitive and thus suffer from the same
imprecision as their flow-insensitive
whole-program counterparts. In the absence of strong updates,
many spurious aliases (such as $*a$ and $*c$) result, 
causing $z$ to point to both $i$ and $u$.
As a result, a false alarm for $z$ is issued,
as discussed earlier.

However, \vfsu performs strong updates flow-sensitively by
filtering out the four spurious pre-computed value-flows marked by \Cross.
As $t3$ points to $d$ only, $\lab_{15}\! \vfedge{b}\!\lab_{16}$
is spurious and not traversed. In addition, a strong update
is enabled at $\lab_{15}: *t3 = v$, rendering
$\lab_{14}\! \vfedge{b}\!\lab_{15}$ and $\lab_{14}\!
\vfedge{d}\!\lab_{15}$ spurious. 
Finally,
$\lab_{5}\! \vfedge{a}\!\lab_{9}$ is refined away 
due to another 
strong update performed at $\ell_9$. Thus, \vfsu has avoided
many spurious aliases 
(e.g., $*a$ and $*c$) introduced flow-insensitively
by pre-analysis, resulting
in $\pts(\lvar{\ell_{16},z})\!=\!\{i\}$ precisely. Thus,
no warning for $z$ is issued.

\item \textbf{\emph{Query-based Precision Control}}

To balance efficiency and precision, \vfsu operates in a 
hybrid multi-stage analysis framework.
When asked to answer the query
$\pts(\lvar{\ell_{16},z})$ under a budget, say, 
a maximum sequence of three steps traversed, \vfsu will stop its
traversal from $\ell_9$ to $\ell_8$ (at \circled{4})
in Figure~\ref{fig:motex}(d) and  
fall back to the pre-computed results by returning 
$\pts({\lvar{\lab_{16},z}}) = \{u,i\}$. In this case, a false
positive for $z$ will end up being reported.

\end{itemize}


\section{DEMAND-DRIVEN STRONG UPDATES}
\label{sec:intra}

\begin{figure*}[t]
\centering
\renewcommand{\arraystretch}{.7}

\begin{tabular}{l}

\begin{tabular}{ll}
\\
\ruledef{\rulename{ADDR}\hspace{-2mm}}{
\begin{tabular}{lll}
$\lab: p = \&o$ 
\end{tabular}
}{
\begin{tabular}{ll}
$ \lvar{\lab,p} \vfreach \allocation{o}$
\end{tabular}
}
\end{tabular}

\quad

\begin{tabular}{ll}
\\
\ruledef{\rulename{COPY}\hspace{-2mm}}{
\begin{tabular}{lll}
$\lab: p = q$ \quad $\lab' \vfedge{q} \lab $
\end{tabular}
}{
\begin{tabular}{ll}
$ \lvar{\lab,p} \vfreach \lvar{\lab',q}$
\end{tabular}
}
\end{tabular}

\\ \\

\begin{tabular}{ll}
\\
\ruledef{\rulename{PHI}\hspace{-2mm}}{
\begin{tabular}{lll}
$\lab: p = \phi(q,r)$\quad $\lab' \vfedge{q} \lab $ \quad $\lab'' \vfedge{r} \lab $
\end{tabular}
}{
\begin{tabular}{ll}
$ \lvar{\lab,p} \vfreach \lvar{\lab',q}$ \quad $\lvar{\lab,p} \vfreach \lvar{\lab'',r}$
\end{tabular}
}
\end{tabular}

\\ \\

\begin{tabular}{ll}
\\
\ruledef{\rulename{Field}\hspace{-2mm}}{
\begin{tabular}{lll}
$\lab: p =\&q\!\rightarrow\!fld$ 
\quad
$\lab' \vfedge{q} \lab $ 
\quad $\lvar{\lab',q} \vfreach \allocation{o}$
\end{tabular}
}{
\begin{tabular}{ll}
$\lvar{\lab,p} \vfreach \allocation{o_{fld}}$
\end{tabular}
}
\end{tabular}

\\ \\

\begin{tabular}{ll}
\\
\ruledef{\rulename{LOAD}\hspace{-2mm}}{
\begin{tabular}{lll}
$\lab: p = *q$ \ \
$\lab'' \vfedge{q} \lab $ \ \
$\lvar{\lab'',q} \vfreach \allocation{o}$
\ \ $\lab' \vfedge{o} \lab $
\end{tabular}
}{
\begin{tabular}{ll}
$ \lvar{\lab,p} \vfreach \lvar{\lab',o}$
\end{tabular}
}
\end{tabular}

\\ \\

\begin{tabular}{ll}
\\
\ruledef{\rulename{STORE}\hspace{-2mm}}{
\begin{tabular}{lll}
$\lab: *p = q$ \ \
$\lab'' \vfedge{p} \lab $ \ \
$\lvar{\lab'',p}\vfreach \allocation{o}$ \ \ $\lab' \vfedge{q} \lab $
\end{tabular}
}{
\begin{tabular}{ll}
$ \lvar{\lab,o} \vfreach \lvar{\lab',q}$ \quad
\end{tabular}
}
\end{tabular}

\\ \\

\begin{tabular}{ll}
\\
\ruledef{\rulename{SU/WU}\hspace{-2mm}}{
\begin{tabular}{lll}
$\lab: *p = \_$ \quad \quad $\lab' \vfedge{o} \lab $\quad \quad $o \in \calA \backslash \ \! \killset(\lab,p)$
\end{tabular}
}{
\begin{tabular}{ll}
$ \lvar{\lab,o} \vfreach \lvar{\lab',o}$ \quad
\end{tabular}
}
\end{tabular}

\\ \\

\begin{tabular}{ll}
\ruledef{\rulename{Call}\hspace{-4mm}}{
\begin{tabular}{lll}
$\lab: \_ = q(\dots,r,\dots) \quad \muop{o_j}$ \qquad\quad  $\lab': f(\dots,r',\dots) \quad \chiop{o_{i+1}}{o_i}$ \\
\qquad\qquad
$\lab'' \vfedge{q} \lab $ \quad
$\lvar{\lab'',q}\vfreach \allocation{o_f}$ \quad $\lab\vfedge{r}\lab'$  \quad $\lab\vfedge{o}\lab'$
\end{tabular}
}{
\begin{tabular}{ll}
$ \lvar{\lab',r'}\vfreach\lvar{\lab,r}$ \quad$ \lvar{\lab',o}\vfreach\lvar{\lab,o}$\\
\end{tabular}
}
\end{tabular}

\\ \\

\begin{tabular}{ll}
\ruledef{\rulename{Ret}\hspace{-4mm}}{
\begin{tabular}{lll}
$\lab: p = q(\dots)  \quad \chiop{o_{j+1}}{o_j}$ \qquad\quad $\lab': ret_f\ p' \quad \muop{o_i}$ \\
\qquad\qquad 
$\lab'' \vfedge{q} \lab $ \quad
$\lvar{\lab'',q}\vfreach \allocation{o_f}$ \quad $\lab' \vfedge{p'} \lab$  \quad $\lab' \vfedge{o} \lab$
\end{tabular}
}{
\begin{tabular}{ll}
$\lvar{\lab,p} \vfreach \lvar{p',\lab'}$ \quad $\lvar{\lab,o} \vfreach \lvar{\lab',o}$\\
\end{tabular}
}
\end{tabular}


\\ \\
\begin{tabular}{ll}
\\
\ruledef{\rulename{Compo}\hspace{-2mm}}{
\begin{tabular}{lll}
$\locvar \vfreach \locvar'$ \hspace{5mm}$\locvar' \vfreach \locvar''$ 
\end{tabular}
}{
\begin{tabular}{ll}
$\locvar \vfreach \locvar''$
\end{tabular}
}
\end{tabular}

\\ \\ \\
\begin{tabular}{c}

\(
  \killset(\lab,p)=\begin{cases}
    \{o'\} & \text{if $\pts(\lvar{\lab,p})\!=\! \{o'\} \wedge o' \in \singletons$}\\
    \calA & \text{else if $\pts(\lvar{\lab,p})\!=\! \varnothing$}\\
    \varnothing & \text{otherwise}
  \end{cases}
\)

\end{tabular}
\\\hline
\end{tabular}

\caption{Single-stage flow-sensitive \vfsu analysis with
	demand-driven strong updates.
	\label{fig:ddafs}}
\end{figure*}

We introduce our demand-driven pointer analysis with
strong updates, as illustrated in
Figure~\ref{fig:framework}.
We first describe our inference rules in a flow-sensitive setting
(Section~\ref{sec:dda}). We then
discuss our context-sensitive extension (Section~\ref{sec:inter}).
Finally, we present our hybrid multi-stage analysis framework
(Section~\ref{sec:sensivity}). All our analyses
are field-sensitive, thereby enabling more
strong updates to be performed to struct objects.

\subsection{Formalism: Flow-Sensitivity}
\label{sec:dda}

We present a formalization of a single-stage \vfsu consisting
of only a flow-sensitive (FS) analysis. Given a program, 
\vfsu will operate on its SVFG representation $G_{\subvfg}$ constructed by
applying Andersen's analysis~\cite{andersen1994program}
as a pre-analysis, as discussed in
Section~\ref{sec:svfg} and illustrated in Section~\ref{sec:mot}.

Let $\locvars\!=\!\calL\times \calV$ be the set of labeled variables
$lv$, where $\calL$ is the set of 
statement labels and $\calV=\calP\cup\calO$ as defined
in Table~\ref{tab:domain}.
\vfsu conducts a backward reachability analysis flow-sensitively on $G_{\subvfg}$
by computing a reachability relation,
$\vfreach \ \subseteq \locvars \times \locvars$. In
our formalism,
$\lvar{\lab,v} \vfreach \lvar{\lab', v'}$ 
signifies a value-flow from a def of $v'$ at $\ell'$ to a use of $v$ at $\ell$ through one
or multiple value-flow paths in $G_{\subvfg}$.  For an 
object $o$ created at an \textsc{AddrOf} statement, i.e.,
an allocation site at $\ell'$, identified as
$\lvar{\ell',o}$, we must distinguish it from $\lvar{\ell,o}$  accessed elsewhere
at $\ell$ in our inference rules. Our abbreviation for $\lvar{\ell',o}$ is $\allocation{o}$.

Given a points-to query $\lvar{\lab,v}$,
\vfsu computes $\pts(\lvar{\lab,v})$, i.e., the points-to set of 
$\lvar{\lab,v}$ by finding all 
reachable target objects $\allocation{o}$, defined as follows:
\begin{equation}
\label{eq:ptr}
\pts(\lvar{\lab,v}) = \{o \mid \lvar{\lab,v} \vfreach \allocation{o}\}
\end{equation}

Despite flow-sensitivity,
our formalization in Figure~\ref{fig:ddafs}
makes no explicit references to program
points. As \vfsu operates on the def-use chains in $G_{\subvfg}$,
each variable $\lvar{\ell,v}$ mentioned in a rule
appears at the point just after $\ell$, where $v$ is defined.

Let us examine our rules in detail.
By \rulename{ADDR}, 
an object $\allocation{o}$ created at an allocation site $\lab$ 
is backward reachable from $p$ at $\ell$ (or precisely at the point
after $\ell$). The pre-computed
direct value-flows across the top-level variables in $G_{\subvfg}$
are always precise
(\rulename{COPY} and \rulename{PHI}). In partial SSA form,
\rulename{PHI} exists only for top-level variables
(Section~\ref{sec:svfg}). 
 
However, the indirect value-flows across the address-taken 
variables in $G_{\subvfg}$ can be imprecise; they  
need to be refined on the fly to remove
the spurious aliases thus introduced.
When handling a load $p=*q$ in \rulename{LOAD},
we can traverse backwards from $p$ at $\ell$ to the def
of $o$ at 
$\lab'$ only if $o$ is \emph{actually} used by, i.e., aliased with $*q$ 
at $\lab$, which requires the reachability relation 
$\lvar{\lab'',q} \vfreach \allocation{o}$ to be computed recursively. 
A store $*p=q$ is handled similarly 
(\rulename{Store}): 
$q$ defined at $\ell'$
can be reached backwards by $o$ at $\lab$ only if $o$ 
is aliased with $*p$
at $\ell$. 

If $*q$ in a load $\cdots=*q$ 
is aliased with $*p$ in a store
$*p=\cdots$ executed earlier, then $p$ and
$q$ must be both backward reachable
from $\allocation{o}$.  Otherwise, any alias relation
established
between $*p$ and $*q$ in $G_{\subvfg}$ by pre-analysis
must be spurious and will thus be filtered out by
value-flow refinement.

\rulename{SU/WU} models strong and weak updates 
at a store $\ell\!:\! p\!=\!\_$. Defining its kill set $\killset(\ell,p)$ 
involves three cases.
In Case~(1), $p$ points to one \emph{singleton object}
$o'$ in $\singletons$, which contains all objects in $\calA$ except the
local variables in recursion, arrays (treated monolithically) or heap
objects~\cite{strongupdate}. In
Section~\ref{sec:inter}, we discuss how to apply
strong updates to heap objects context-sensitively.
A strong update is then possible to $o$.
By killing its old contents at $\lab'$, no further backward
traversal along the def-use chain $\lab' \vfedge{o} \lab$ is needed.
Thus, $\lvar{\lab,o} \vfreach \lvar{\lab', o}$ is falsified. 
In Case (2), the points-to set of $p$ is empty. Again, further
traversal to $\lvar{\lab',o}$ must be prevented to 
avoid dereferencing a null pointer as is 
standard~\cite{hardekopf2009semi,hardekopfflow,strongupdate}. In
Case (3), a weak update is performed to $o$ so that
its old contents at $\ell'$ are preserved. Thus, 
$\lvar{\lab,o} \vfreach \lvar{\lab', o}$ is
established, which implies that 
the backward traversal along $\lab' \vfedge{o} \lab$ must continue.

\rulename{Field} handles field-sensitivity. For a field access (e.g., $p=\&q\!\rightarrow\!fld$), pointer $p$ points to the field object $o_{fld}$ of object $o$ pointed to by $q$.

\rulename{Call} and \rulename{Ret} handle the
reachability traversal interprocedurally by computing the 
call graph for the program
on the fly instead of relying on the imprecisely
pre-computed call graph built by the pre-analysis as 
in~\cite{hardekopfflow}. In the SVFG,
the interprocedural value-flows sinking into a callee 
function $f$ may come from a spurious indirect callsite 
$\lab$. To avoid this, 
both rules ensure that the function pointer $q$ at $\lab$ actually
points to $f$ (\rulename{Call} and \rulename{Ret}). 
Essentially, given a points-to query $z$ at an 
indirect callsite $\lab:z=(*fp)()$. Instead of analyzing all the callees 
found by the pre-analysis, \vfsu recursively computes
the points-to set of $fp$ to discover new callees at the
callsite and then continues refining $\pts(\lvar{\lab,z})$ using the new callees.

\begin{sidewaysfigure}
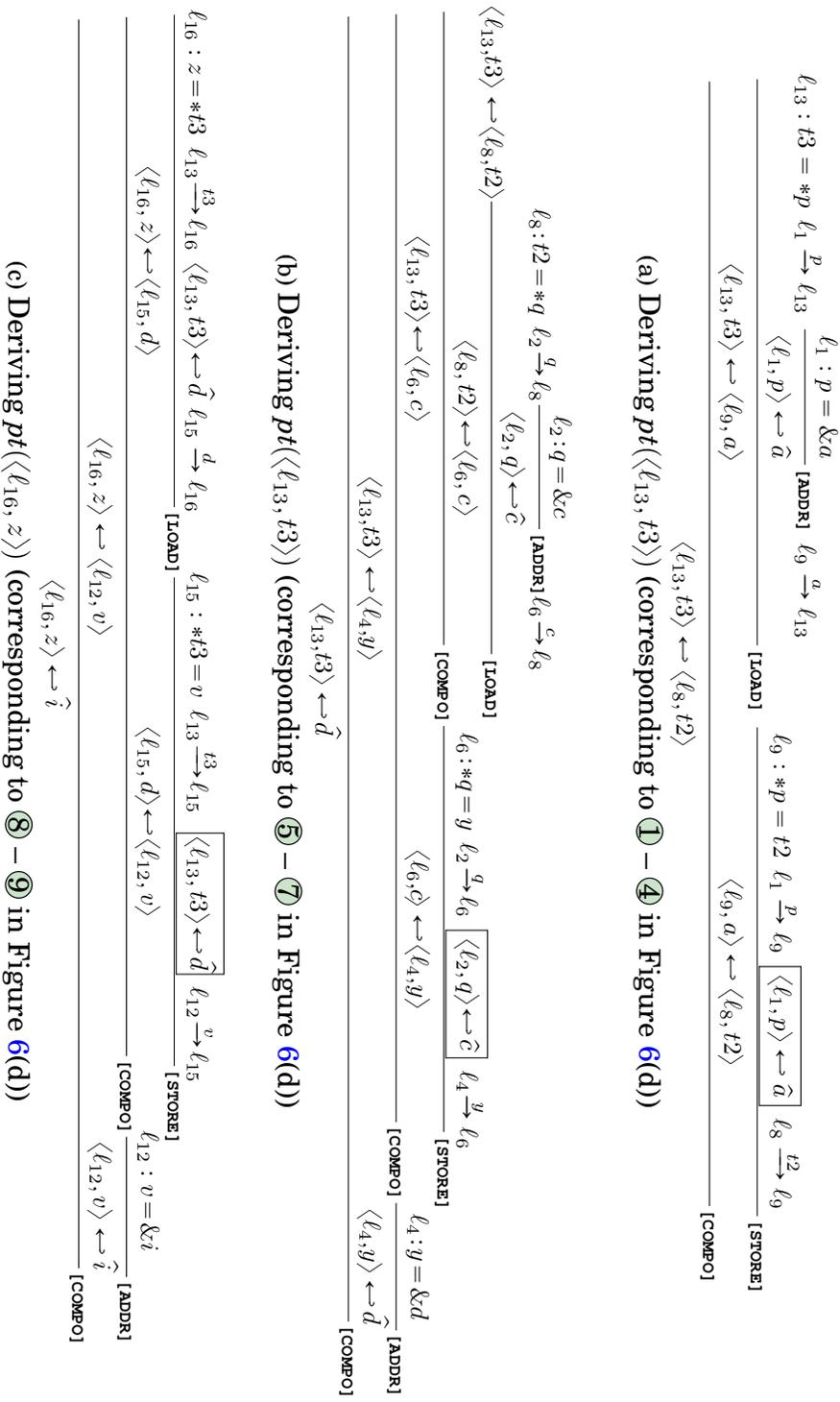

    \centering
\newcommand{\precond}[1]{\begin{array}[c]{c}#1\end{array}}
\newcommand{\precondX}[2]{\begin{array}[c]{cc} #1 & #2 \end{array}}
\newcommand{\postcond}[1]{\begin{array}[c]{c}#1\end{array}}
\newcommand{\rulenamefont}[1]{\scalebox{0.78}{\textbf{\rulename{#1}}}}
\newcommand{\infer}[3]{\frac{\precond{#1}}{\postcond{#2}}#3}

\begin{tabular}{c}
\def\arraystretch{1.6}
$
\infer{
  \hspace{-1ex}
  \infer{
    \hspace{-1ex}
    \lab_{13}: t3 = *p  
    \hspace{1ex}
    \lab_{1} \vfedge{p} \lab_{13}
    \hspace{1ex}
    \infer{
      \lab_{1}: p = \&a
    }{
      \lvar{\lab_{1}, p} \vfreach \allocation{a}
    }{
      \rulenamefont{ADdr}
    }
    \hspace{1ex}
    \lab_{9} \vfedge{a} \lab_{13}
  }{
    \lvar{\lab_{13}, t3} \vfreach \lvar{\lab_{9}, a}
  }{
    \rulenamefont{LoaD}
  }
  \hspace{1ex}
  \infer{
    \lab_{9}: *p = t2
    \hspace{1ex}
    \lab_{1} \vfedge{p} \lab_{9}
    \hspace{1ex}
    \fbox{$\lvar{\lab_{1}, p} \vfreach \allocation{a}$}
    \hspace{1ex}
    \lab_{8} \vfedge{t2} \lab_{9}
  }{
    \lvar{\lab_{9}, a} \vfreach \lvar{\lab_{8}, t2}
  }{
    \rulenamefont{STore}
    \hspace*{-8ex}
  }
}{
  \lvar{\lab_{13}, t3} \vfreach \lvar{\lab_{8},t2}
}{
  \rulenamefont{COMpo}
}
$
\\
\vspace*{.5ex}
(a) \scalebox{1.15}{Deriving $\pts(\lvar{\lab_{13},t3})$ (corresponding to \circled{1} -- \circled{4} in Figure~\ref{fig:motex}(d))}
\\
\vspace*{1ex}
\\
\hspace{6mm}
\def\arraystretch{1.6}
$
\infer{
  \hspace{-1ex}
  \infer{
    \hspace{-1ex}
    \infer{
      \hspace{-1ex}
      \lvar{\lab_{13},\! t3} \vfreach\! \lvar{\lab_{8},\! t2}
      \infer{
        \lab_{8}\!:\! t2\! =\! *q
        \hspace{1ex}
        \lab_{2}\! \vfedge{q}\! \lab_{8}
        \infer{
          \lab_{2}\!:\! q\! =\! \&c
        }{
          \lvar{\lab_{2}, q}\! \vfreach\! \allocation{c}
        }{
          \rulenamefont{ADDR}
        }
        \lab_{6}\! \vfedge{c}\! \lab_{8}
        \hspace*{-2ex}
      }{
        \lvar{\lab_{8}, t2}\! \vfreach\! \lvar{\lab_{6}, c}
      }{
        \rulenamefont{LOAD}
        \hspace*{-7ex}
      }
    }{
      \lvar{\lab_{13}, t3}\! \vfreach\! \lvar{\lab_{6}, c}
    }{
      \rulenamefont{COMPO}
    }
    \infer{
      \lab_{6}\!:\! *q\! =\! y
      \hspace{1ex}
      \lab_{2}\! \vfedge{q}\! \lab_{6}
      \hspace{1ex}
      \fbox{$\lvar{\lab_{2}, q}\! \vfreach\! \allocation{c}$}
      \hspace{1ex}
      \lab_{4}\! \vfedge{y\!} \lab_{6}
      \hspace*{-2ex}
    }{
      \lvar{\lab_{6},\!c\!} \vfreach\! \lvar{\lab_{4},\! y}
    }{
      \rulenamefont{STORE}
      \hspace*{-8ex}
    }
  }{
    \lvar{\lab_{13},\! t3} \vfreach\! \lvar{\lab_{4},\! y}
  }{
    \rulenamefont{COMPO}
  }
  \infer{
    \lab_{4}\!:\! y\! =\! \&d
  }{
    \lvar{\lab_{4},\! y} \vfreach\! \allocation{d}
  }{
    \rulenamefont{ADDR}
    \hspace*{-7ex}
  }
}{
  \lvar{\lab_{13},\! t3} \vfreach\! \allocation{d}
}{
  \rulenamefont{COMPO}
}
$
\\
\vspace*{.5ex}
(b) \scalebox{1.15}{Deriving $\pts(\lvar{\lab_{13},t3})$ (corresponding to \circled{5} -- \circled{7} in Figure~\ref{fig:motex}(d))}
\\
\vspace*{1ex}
\\
\def\arraystretch{1.6}
$
\infer{
  \hspace{-1ex}
  \infer{
    \hspace{-1ex}
    \infer{
      \hspace*{-1ex}
      \lab_{16}: z\! =\! *t3
      \hspace{1ex}
      \lab_{13}\! \vfedge{t3}\! \lab_{16}
      \hspace{1ex}
      \lvar{\lab_{13}, t3}\! \vfreach\! \allocation{d}
      \hspace{1ex}
      \lab_{15} \vfedge{d} \lab_{16}\hspace*{-1ex}
    }{
      \hspace*{-2ex}
      \lvar{\lab_{16}, z}\! \vfreach\! \lvar{\lab_{15}, d}
      \hspace*{-2ex}
    }{
      \rulenamefont{LoaD}\hspace*{-1ex}
    }
    \hspace{1ex}
    \infer{
      \hspace*{-1ex}
      \lab_{15}: *t3\! = \!v
      \hspace{1ex}
      \lab_{13} \!\vfedge{t3} \!\lab_{15} \
      \hspace{1ex}
      \fbox{\!$\lvar{\lab_{13}, t3} \!\vfreach \!\allocation{d}$\!}
      \hspace{1ex}
      \lab_{12} \!\vfedge{v} \!\lab_{15}
      \hspace*{-2ex}
    }{
	  \hspace*{-2ex}
      \lvar{\lab_{15}, d}\! \vfreach\! \lvar{\lab_{12}, v}
      \hspace*{-2ex}
    }{
      \rulenamefont{STore}
      \hspace*{-8ex}
    }
  }{
    \hspace*{-2ex}
    \lvar{\lab_{16}, z} \vfreach \lvar{\lab_{12}, v}
    \hspace*{-2ex}
  }{
    \rulenamefont{COMpo}\hspace*{-1ex}
  }
  \hspace{1ex}
  \infer{
    \hspace*{-7ex}
    \lab_{12}: v\! =\! \&i
    \hspace*{-7ex}
  }{
    \hspace*{-.1ex}
    \lvar{\lab_{12}, v} \vfreach \allocation{i}
  }{
    \rulenamefont{ADdr}
    \hspace*{-7ex}
  }
}{
  \lvar{\lab_{16}, z} \vfreach \allocation{i}
}{
  \rulenamefont{COMpo}
}
$
\\
\vspace*{.5ex}
(c) \scalebox{1.15}{Deriving $\pts(\lvar{\lab_{16},z})$ 
(corresponding to \circled{8} -- \circled{9} in Figure~\ref{fig:motex}(d))}
\\
\vspace*{1ex}
\end{tabular}

\caption{
Reachability derivations for $\pts(\lvar{\ell_{16},z})$ shown in
Figure~\ref{fig:motex}(d) (with reuse of cached
points-to results inside each box).  \label{fig:mottrace}
}
    \label{fig:awesome_image}
\end{sidewaysfigure}

Finally, $\vfreach$ is transitive, stated by \rulename{COMPO}.

Let us try all our rules, by first revisiting our motivating example
where strong updates are performed 
(Example~\ref{ex:mot}) and then examining weak updates
(Example~\ref{ex:wu}).

\begin{example}\rm
Figure~\ref{fig:mottrace} shows how we apply the rules of
\vfsu 
to answer $\pts(\lvar{\lab_{16},z})$ illustrated in
Figure~\ref{fig:motex}(d).
\rulename{SU/WU} (implicit in these derivations) is 
applied to $\ell_6$, $\ell_9$ and $\ell_{15}$
to cause a strong update at each store.
At $\ell_{6}$, $\pts(\lvar{\lab_6,q})=\{c\}$, the old
contents in $c$ are killed.
At $\ell_9$, $\lab_5 \!\vfedge{a}\! \lab_9$ becomes spurious since
$\lvar{\lab_9,a}\! \vfreach\! \lvar{\lab_{5},a}$ is falsified.
At $\ell_{15}$,
$\lab_{14} \!\vfedge{b} \!\lab_{15}$ and $\lab_{14} \!\vfedge{d} \!\lab_{15}$ 
are filtered out since 
$\lvar{\lab_{15},b} {\vfreach} \lvar{\lab_{14},b}$ and 
$\lvar{\lab_{15},d}\! \vfreach\! \lvar{\lab_{14},d}$ are 
falsified.  Finally, $\lab_{15} \vfedge{b} \lab_{16}$ is 
ignored since $t3$ points to $d$ only (\rulename{LOAD}). 
\hfill {\small $\Box$}
\label{ex:mot}
\end{example}

\vfsu improves performance 
by caching points-to results to reduce redundant 
traversal, with
reuse happening in the
marked boxes in Figure~\ref{fig:mottrace}.
For example, in Figure~\ref{fig:mottrace}(c),
$\pts(\lvar{\lab_{13},t3})=\{\allocation{d}\}$ computed
in \rulename{LOAD} is reused in \rulename{STORE}.

\begin{figure}[!h]
\centering
\begin{tabular}{l}
\hspace{1mm}
\hfil\Large\textbf{\textcolor{blue}{}}\hfil\includegraphics[scale=0.5]{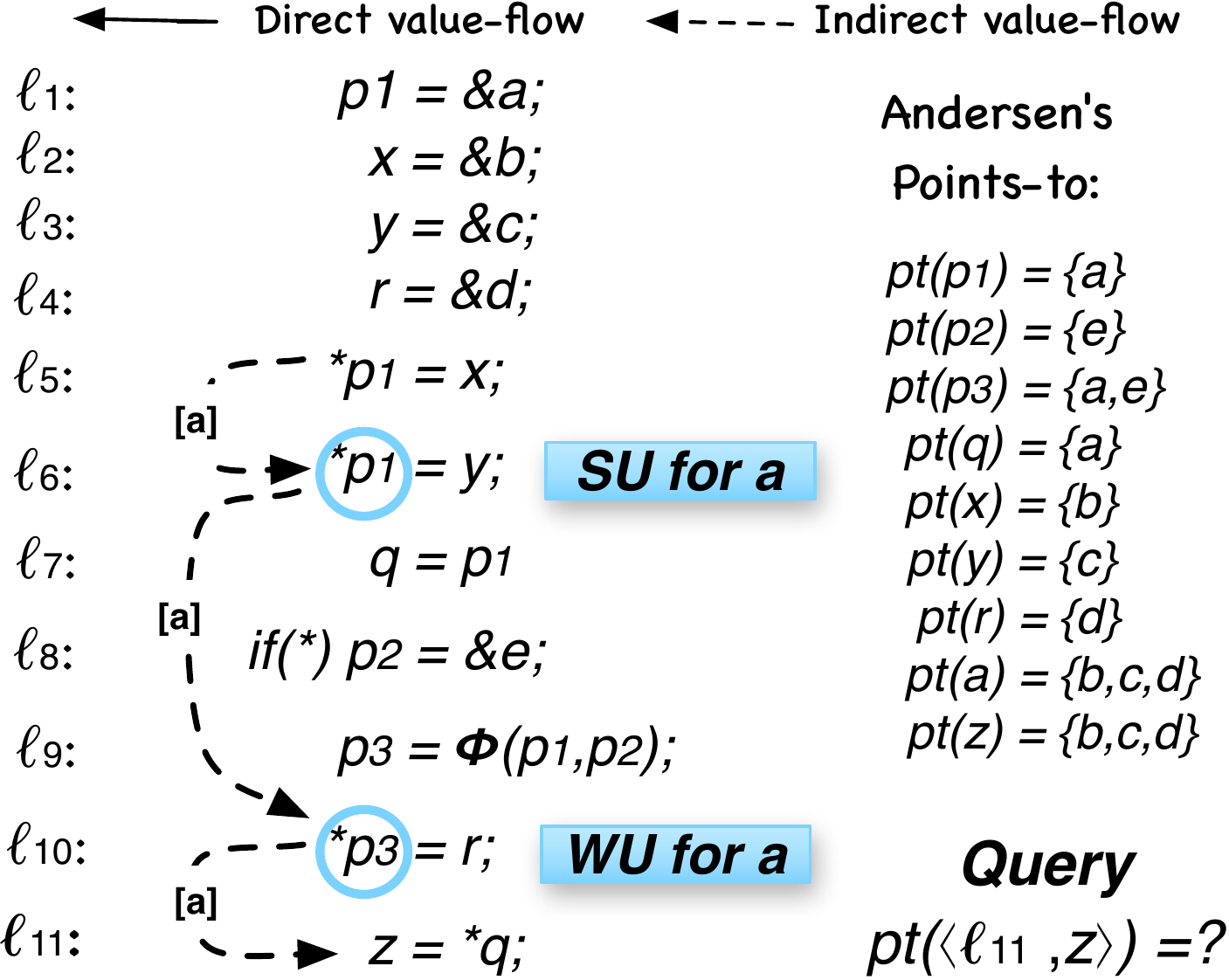}
\\ \hline
\end{tabular} 
\caption{
Resolving $\pts({\lvar{\lab_{11},z}})=\{c,d\}$ with a weak update. 
	\label{fig:wu} }
\end{figure}

\begin{example} \rm
Let us consider a weak update example in Figure~\ref{fig:wu} by 
computing $\pts(\lvar{\ell_{11},z})$ on-demand.
At the confluence point $\lab_9$, $p3$ 
receives the points-to information from 
both $p1$ and $p2$ in its
two branches: $\lvar{\ell_9,p_3}
\vfreach\allocation{a}$ and $\lvar{\ell_9,p_3}\vfreach
\allocation{e}$. Thus, a weak update is performed
to the two locations ${a}$ and ${e}$
at $\lab_{10}$. Let us focus on 
$\allocation{a}$ only. By applying \rulename{STORE},
$\lvar{\ell_{10},a} \vfreach 
\lvar{\ell_4,r} \vfreach \allocation{d}$.
By applying \rulename{SU/WU}, 
$\lvar{\lab_{10},a}\vfreach\lvar{\lab_{6},a}\vfreach\lvar{\lab_{3},y}
\vfreach\allocation{c}$. 
Thus, $\pts(\lvar{\ell_{11},a})=\{c,d\}$, which excludes
$b$ due to a strong update performed at $\ell_6$.
As
$\pts(\lvar{\ell_{7},q})=\{a\}$, we obtain
$\pts(\lvar{\ell_{11},z})=\{c,d\}$.
\label{ex:wu}
\hfill {\small $\Box$}
\end{example}

\begin{figure}[t]
\centering
\begin{tabular}{l}
\hspace{1mm}
\hfil\Large\textbf{\textcolor{blue}{}}\hfil\includegraphics[scale=0.55]{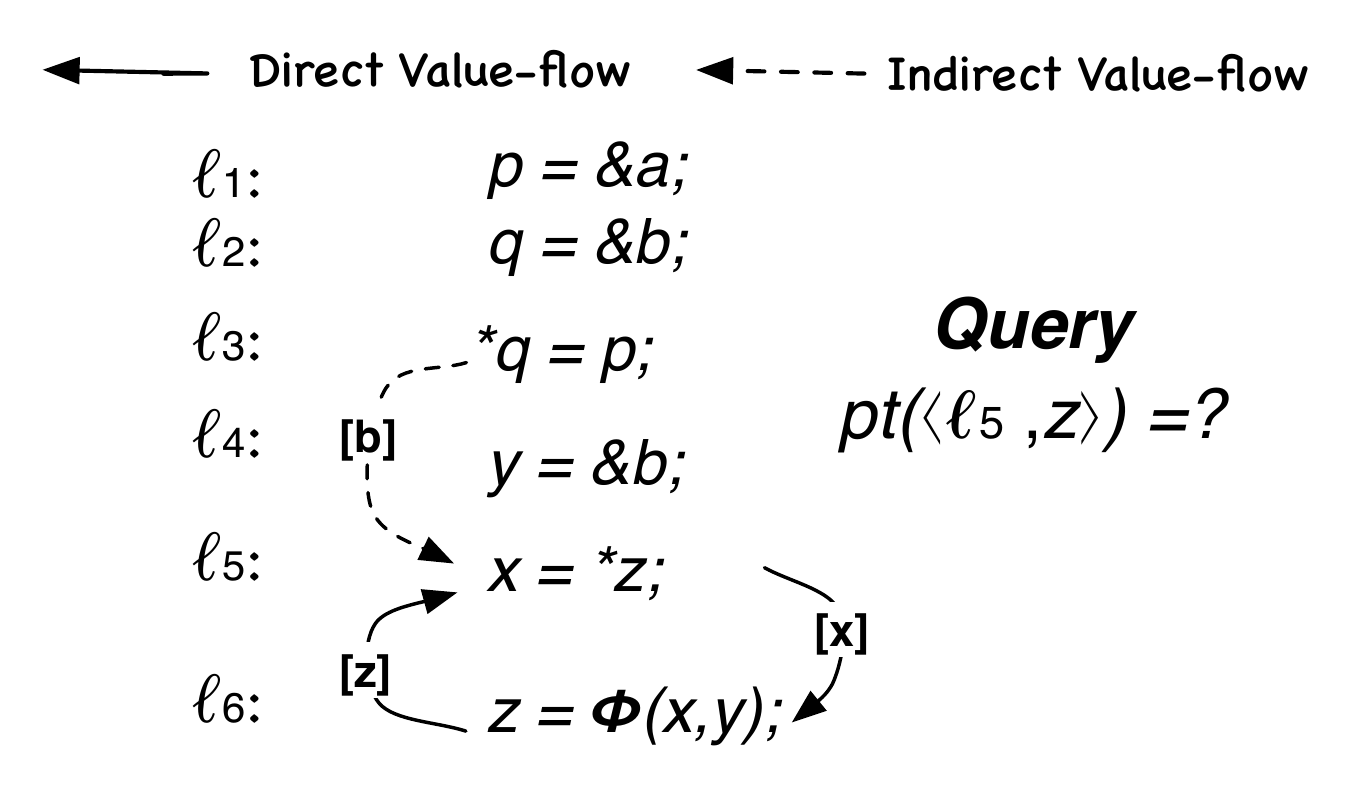}
\\ \hline
\end{tabular} 
\caption{
	Resolving $\pts(\lvar{\lab_{5}, z})\! =\! \{a, b\}$ in a value-flow cycle.
	\label{fig:vfcycle} }
\end{figure}

Unlike~\cite{strongupdate}, which falls back 
to the flow-insensitive points-to information for all weakly updated
objects, \vfsu handles them as precisely as
(whole-program) flow-sensitive analysis subject to a
sufficient budget. In Figure~\ref{fig:wu}, due to 
a weak update performed to $a$ at $\lab_{10}$, 
\pts$(\lvar{\lab_{10},a})=\{c,d\}$ is obtained, forcing their approach
to adopt $\pts(\lvar{\lab_{10},a}) =\{b,c,d\}$ thereafter, causing
\pts$(\lvar{\lab_{11},z}) =\{b,c,d\}$. By maintaining
flow-sensitivity with a strong
update applied to $\ell_6$ to kill $b$, \vfsu obtains
\pts$(\lvar{\lab_{11},z}) = \{c,d\}$ precisely.

\subsubsection{\textbf{Handling Value-Flow Cycles}}

To compute soundly and precisely the points-to information in
a value-flow cycle in the SVFG,
\vfsu retraverses it whenever new points-to information is
found until a fix point is reached.

\begin{example}\rm
\label{ex:vfcycle}
Figure~\ref{fig:vfcycle} shows 
a value-flow cycle formed by
$\lab_5\vfedge{x}\lab_{6}$ and $\lab_{6}\vfedge{z}\lab_{5}$. 
To compute \pts$(\lvar{\lab_{6},z})$, we must
compute \pts$(\lvar{\lab_{5},x})$, which
requires the aliases of $*z$ 
at the load $\lab_{5}: x = *z$ to be found by using
\pts$(\lvar{\lab_{6},z})$.  \vfsu computes \pts$(\lvar{\lab_{6},z})$
by analyzing this value-flow cycle in two iterations. In the first iteration,
a pointed-to target $\allocation{b}$ is found since
$\lvar{\lab_{6},z}\vfreach\lvar{\lab_{4},y}\vfreach\allocation{b}$.
Due to $\lvar{\lab_{2},q}\vfreach\allocation{b}$, $*z$ and $*q$ are found to be
aliases. In the
second iteration, another target $\allocation{a}$ is found since
$\lvar{\lab_{6},z}\vfreach\lvar{\lab_{5},x}\vfreach\lvar{\lab_{3},b}\vfreach\lvar{\lab_{1},p}\vfreach\allocation{a}$.
Thus, \pts$(\lvar{\lab_{6}, z}) = \{a, b\}$ is obtained.
\hfill {\small $\Box$}
\end{example}

\subsubsection{\textbf{Field-Sensitivity}}

Field-insensitive pointer analysis does not
distinguish different fields of a struct object, and 
consequently, gives up opportunities for performing strong updates to 
a struct object, as a struct object may actually
represent its distinct fields. In contrast, 
\vfsu is truly field-sensitive, by avoiding the two
limitations altogether.

\begin{figure}[h]
\centering
\begin{tabular}{l}
\hfil\Large\textbf{\textcolor{blue}{}}\hfil\includegraphics[scale=0.62]{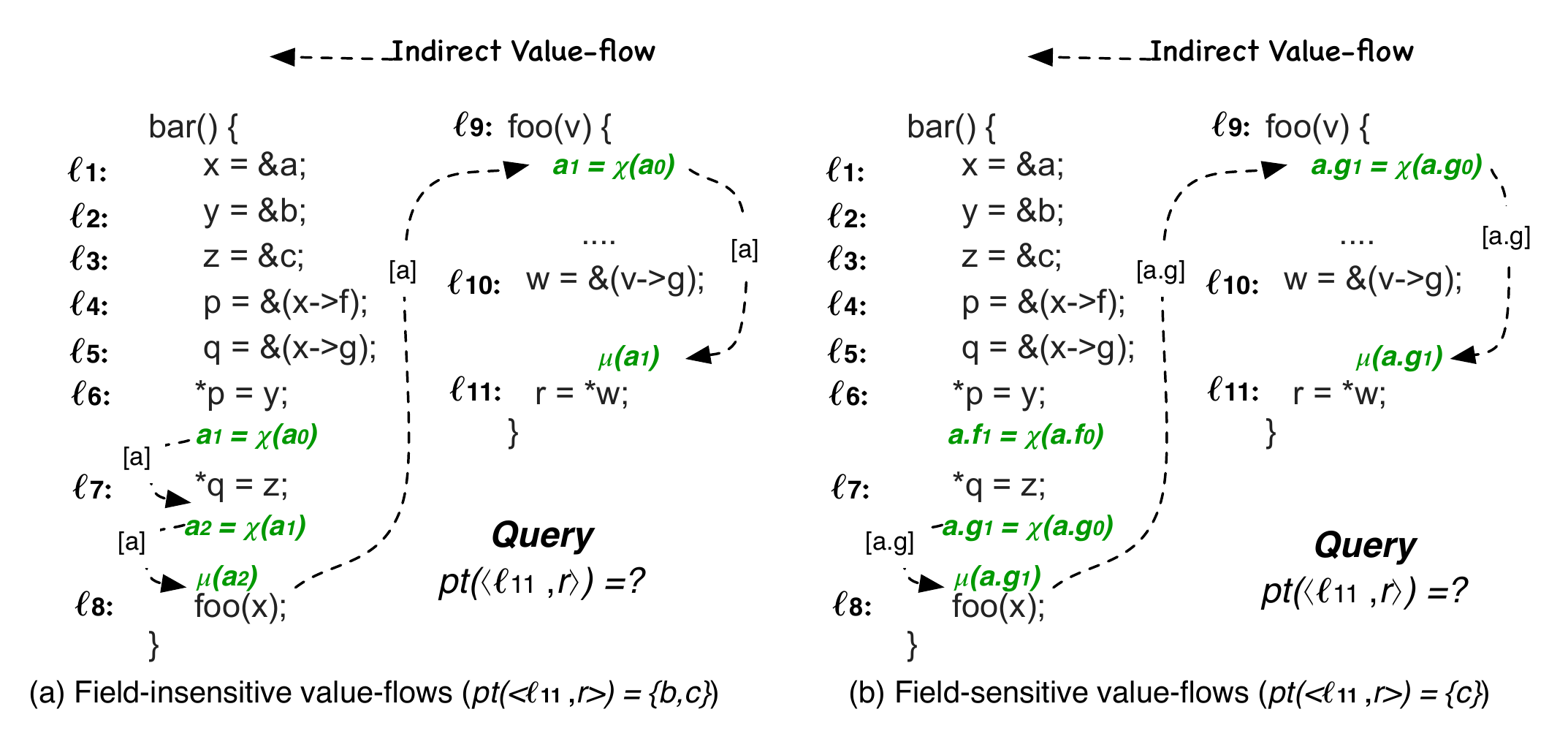}
\\ \hline
\end{tabular} 
\caption{
	Resolving $\pts(\lvar{\lab_{11}, r})\! =\! \{c\}$
	with field-sensitivity.
	\label{fig:field} 
	}
\end{figure}

\begin{example}\rm
\label{ex:field}
Figure~\ref{fig:field} illustrates the effects of
field-sensitivity on computing the points-to information
for $r$ at $\lab_{11}$. Without field-sensitivity, as
illustrated in Figure~\ref{fig:field}(a), the
two statements at $\lab_4$ and $\lab_5$ are analyzed as if
they were $\lab_4: p = \&x$ and $\lab_5: q = \&x$. As
a result, no strong update is possible at $\lab_6$ and
$\lab_7$, since $x$, which represents possibly multiple fields,
is not a singleton. Thus,
$\pts(\lvar{\lab_{11},r})=\{b,c\}$.

\vfsu is field-sensitive.
To answer the points-to query for $r$ at $\lab_{11}$, we
compute first
$\lvar{\lab_{11},r}\vfreach\lvar{\lab_{10},w}$ and then
$\lvar{\lab_{10},v}\vfreach\lvar{\lab_9,v}\vfreach\lvar{\lab_8,x}\vfreach\lvar{\lab_1,x}\vfreach{\allocation{a}}$. By
applying \rulename{Field} at $\lab_{10}$ and \rulename{Load}
at $\lab_{11}$, we obtain
$\lvar{\lab_{11},r}\vfreach \lvar{\lab_{11},a.g}$. By traversing
the three indirect
def-use chains for $a.g$, $\lab_7\vfedge{b.g}\lab_8$, $\lab_8\vfedge{a.g}\lab_9$ and $\lab_9\vfedge{a.g}\lab_{11}$,
backwards from $\lab_{11}$, we obtain $\pts(\lvar{\lab_{11},r}) 
\vfreach\lvar{\lab_9,a.g}\vfreach\lvar{\lab_8,a.g}\vfreach\lvar{\lab_7,a.g}\vfreach
{\lvar{\lab_3,z}} \vfreach{\allocation{c}}$.
\hfill {\small $\Box$}
\end{example}

\subsubsection{Properties \label{sec:prop}}

\begin{theorem}[Soundness] 
\label{thm:sound}
\vfsu is sound in analyzing a program as long as its pre-analysis (for computing the SVFG of the program) is sound.
\begin{proof}
When building the SVFG for a program,
the def-use chains for its top-level variables are 
identified explicitly in its partial SSA form. 
If the pre-analysis (for computing the sparse value-flow
graph of the program)
is sound, then the def-use chains built for all
the address-taken variables are over-approximate. 
According to its inference rules in Figure~4,
\vfsu performs essentially a flow-sensitive analysis 
on-demand, by restricting the propagation of points-to
information along the precomputed def-use chains, and falls back to
the sound points-to information computed by the pre-analysis when
running out of its given budgets. Thus, \vfsu is sound if
the pre-analysis is sound.
\end{proof}
\end{theorem}

\begin{theorem}[Precision] 
\label{thm:prec}
Given a points-to query $\lvar{\lab,v}$, 
$\pts(\lvar{\lab,v})$ computed by \vfsu is 
the same as that computed by (whole-program) \FS if \vfsu
can successfully resolve the points-to query within a given budget. 
\begin{proof}
Let $\pts_{\vfsu}(\lvar{\lab,v})$ and $\pts_{\FS}(\lvar{\lab,v})$ be the points-to
sets computed by \vfsu and \FS, respectively. By 
Theorem~1,
$\pts_{\vfsu}(\lvar{\lab,v})\supseteq\pts_{\FS}(\lvar{\lab,v})$, since
\vfsu is a demand-driven version of \FS and thus
cannot be more precise. To show that
$\pts_{\vfsu}(\lvar{\lab,v})\subseteq\pts_{\FS}(\lvar{\lab,v})$, we note that \vfsu
operates on the SVFG of the program to improve its efficiency, by
also filtering out 
value-flows imprecisely pre-computed by the pre-analysis.
For the
top-level
variables, their direct value-flows are precise. So \vfsu proceeds
exactly the same as \FS (\rulename{ADDR}, \rulename{COPY}, 
\rulename{PHI}, \rulename{Field}, \rulename{Call}, \rulename{Ret} and \rulename{COMPO}). For the
address-taken variables, \vfsu establishes the same indirect
value-flows flow-sensitively as \FS does but in a demand-driven manner,
by refining away imprecisely pre-computed value-flows
(\rulename{LOAD}, \rulename{STORE}, \rulename{SU/WU},  \rulename{Call}, \rulename{Ret} and 
\rulename{COMPO}). If \vfsu can complete its query within 
the given budget, then
$\pts_{\vfsu}(\lvar{\lab,v})\subseteq\pts_{\FS}(\lvar{\lab,v})$.
Thus, $\pts_{\vfsu}(\lvar{\lab,v})=\pts_{\FS}(\lvar{\lab,v})$.
\end{proof}
\end{theorem}

\begin{figure*}[t]
	\centering
\scalebox{1}{
\hspace{-4mm}
\renewcommand{\arraystretch}{1}

\begin{tabular}{l}

\begin{tabular}{ll}
\hspace{-3mm}
\begin{tabular}{ll}
\ruledef{\rulename{C-ADDR}\hspace{-4mm}}{
\begin{tabular}{lll}
$\cxt,\lab: p = \&o$ 
\end{tabular}
}{
\begin{tabular}{ll}
$\lvar{\cxt, \lab, p} \vfreach (\cxt, \allocation{o})$
\end{tabular}
}
\end{tabular}
\begin{tabular}{ll}
\ruledef{\rulename{C-COPY}\hspace{-4mm}}{
\begin{tabular}{lll}
$\cxt,\lab: p = q$ \quad $ \lab' \vfedge{q} \lab $
\end{tabular}
}{
\begin{tabular}{ll}
$ \lvar{\cxt,\lab,p} \vfreach \lvar{\cxt,\lab',q}$
\end{tabular}
}
\end{tabular}

\end{tabular}

\\ \\  
\begin{tabular}{ll}
\ruledef{\rulename{C-PHI}\hspace{-4mm}}{
\begin{tabular}{lll}
$\cxt,\lab: p = \phi(q,r)$ \quad $ \lab' \vfedge{q} \lab $ \quad $ \lab'' \vfedge{r} \lab $
\end{tabular}
}{
\begin{tabular}{ll}
$ \lvar{\cxt,\lab,p} \vfreach \lvar{\cxt,\lab',q}$ \quad $\lvar{\cxt,\lab,p} \vfreach \lvar{\cxt,\lab'',r}$
\end{tabular}
}
\end{tabular}

\\ \\ 

\begin{tabular}{ll}
\ruledef{\rulename{C-Field}\hspace{-2mm}}{
\begin{tabular}{lll}
$\cxt,\lab: p =\&q\!\rightarrow\!fld$ 
\quad $ \lab' \vfedge{q} \lab $
\quad $\lvar{\cxt,\lab',q} \vfreach (\cxt',\allocation{o})$
\end{tabular}
}{
\begin{tabular}{ll}
$\lvar{\cxt,\lab,p} \vfreach(\cxt', \allocation{o.fld})$
\end{tabular}
}
\end{tabular}

\\ \\ 

\begin{tabular}{ll}
\ruledef{\rulename{C-LOAD}\hspace{-4mm}}{
\begin{tabular}{lll}
$\cxt,\lab: p = *q$
\ \ $ \lab'' \vfedge{q} \lab $
\ \ $\lvar{\cxt,\lab'',q} \vfreach (\cxt',\allocation{o})$
\ \ $ \lab' \vfedge{o} \lab $
\end{tabular}
}{
\begin{tabular}{ll}
$\lvar{\cxt,\lab,p} \vfreach \lvar{\cxt',\lab',o}$
\end{tabular}
}
\end{tabular}

\\ \\ 
\begin{tabular}{ll}
\ruledef{\rulename{C-Store}\hspace{-4mm}}{
\begin{tabular}{lll}
$\cxt,\lab: *p = q$ \ 
 \ $ \lab'' \vfedge{p} \lab $
\ \
$\lvar{\cxt,\lab'',p} \vfreach (\cxt',\allocation{o})$
\ \ $ \lab' \vfedge{q} \lab $
\end{tabular}
}{
\begin{tabular}{ll}
$ \lvar{\cxt',\lab,o} \vfreach \lvar{\cxt,\lab',q}$ \\
\end{tabular}
}
\end{tabular}

\\ \\ 

\begin{tabular}{ll}
\ruledef{\rulename{C-SU/WU}\hspace{-4mm}}{
\begin{tabular}{lll}
$\cxt,\lab: *p = \_$ \quad $ \lab' \vfedge{o} \lab $  \quad \quad $(\cxt',o) \in \calA \backslash \ \! \killset(\cxt,\lab,p)$
\end{tabular}
}{
\begin{tabular}{ll}
$ \lvar{\cxt',\lab,o} \vfreach \lvar{\cxt',\lab',o}$
\end{tabular}
}
\end{tabular}

\\ \\ 
\begin{tabular}{ll}
\ruledef{\rulename{C-Compo}\hspace{-4mm}}{
\begin{tabular}{lll}
$\locvar \vfreach \locvar'$ \hspace{1mm}$\locvar' \vfreach \locvar''$ 
\end{tabular}
}{
\begin{tabular}{ll}
$\locvar \vfreach \locvar''$
\end{tabular}
}
\end{tabular}

\\ \\ 
\begin{tabular}{ll}
\ruledef{\rulename{C-Call}\hspace{-4mm}}{
\begin{tabular}{lll}
$\cxt,\lab: \_ = q(\dots,r,\dots) \quad \muop{o_j}$ 
\qquad\quad $ \lab'' \vfedge{q} \lab $ \
$\lvar{\cxt,\lab'',q}\vfreach (\_,\allocation{o_f})$ \ $\lab\vfedge{r}\lab'$  \ $\lab\vfedge{o}\lab'$   \\
$\cxt',\lab': f(\dots,r',\dots) \quad \chiop{o_{i+1}}{o_i}$ \qquad\quad $\cxt = \cxt'\cpop\lab$ 
\end{tabular}
}{
\begin{tabular}{ll}
$ \lvar{\cxt',\lab',r'}\vfreach\lvar{\cxt,\lab,r}$ \quad$ \lvar{\cxt',\lab',o}\vfreach\lvar{\cxt,\lab,o}$\\
\end{tabular}
}
\end{tabular}

\\ \\ 
\begin{tabular}{ll}
\ruledef{\rulename{C-Ret}\hspace{-4mm}}{
\begin{tabular}{lll}
$\cxt,\lab: p = q(\dots) \ \chiop{o_{j+1}}{o_j}$ 
\qquad\quad $ \lab'' \vfedge{q} \lab $ 
\ $\lvar{\cxt,\lab'',q}\vfreach (\_,\allocation{o_f})$ \ $\lab' \vfedge{p'} \lab$  \ $\lab' \vfedge{o} \lab$ \\ $\cxt', \lab': ret_f\ p' \quad \muop{o_i}$ \qquad\qquad\qquad\quad $\cxt'=\cxt\cpush\lab$
\end{tabular}
}{
\begin{tabular}{ll}
$\lvar{\cxt,\lab,p} \vfreach \lvar{\cxt',p',\lab'}$ \quad $\lvar{\cxt,\lab,o} \vfreach \lvar{\cxt',\lab',o}$\\
\end{tabular}
}
\end{tabular}


\\ \\ 
\begin{tabular}{c}

\(
  \killset(\cxt,\lab,p)=\begin{cases}
    \{(\cxt', o')\} & \text{if $\pts(\lvar{\cxt,\lab,p})\!=\! \{(\cxt',o')\} \wedge (c',o') \in \cxtsingletons$}\\
    \calA & \text{else if $\pts(\lvar{\cxt,\lab,p})\!=\! \varnothing$}\\
    \varnothing & \text{otherwise}
  \end{cases}
\)

\end{tabular}

\\\hline
\end{tabular}

}
\caption{Single-stage flow- and context-sensitive \vfsu  analysis
with demand-driven strong updates.
	\label{fig:ddacxt}}
\end{figure*}

\subsection{Formalism: Flow- and Context-Sensitivity}
\label{sec:inter}

We extend our flow-sensitive formalization 
by considering also context-sensitivity to enable more 
strong
updates (especially now for heap objects). 
We solve a \emph{balanced-parentheses} problem 
by matching calls and returns to filter out 
unrealizable inter-procedural paths
\cite{Lu13,Reps:1995:PID,Shang12,Sridharan:2006,yan2011demand}.
A context stack $\cxt$ is encoded as a sequence of callsites, 
[$\cc_{1} \dots \cc_{m}$], where $\cc_i$ is a call instruction $\lab$.
$\cxt \cpush \cc$ denotes an operation for pushing a callsite $\cc$ 
into $\cxt$. $\cxt \cpop \cc$ pops
$\cc$ from $\cxt$ if $\cxt$ contains $\cc$ as its top value 
or is empty since a 
realizable path may start and end in different functions.

With context-sensitivity, 
a statement is parameterized additionally
by a context $\cxt$, e.g., $\cxt,\lab\!:\!p\!=\!\&o$,
to represent its instance when its containing
function is analyzed under $\cxt$. 
A labeled variable $lv$ has the form
$\lvar{\cxt,\lab,v}$, representing variable 
$v$ accessed at statement $\lab$
under context $\cxt$.
An object $\allocation{o}$ that is created at an \textsc{AddrOf} statement under context $\cxt$
is also context-sensitive, identified as $(\cxt,\allocation{o})$.

Given a points-to query $\lvar{\cxt,\lab,v}$,
\vfsu computes its points-to set both flow- and context-sensitively 
by applying the rules given in Figure~\ref{fig:ddacxt}:
\begin{equation}
	\pts(\lvar{\cxt,\lab,v}) = \{(\cxt',o) \mid \lvar{\cxt,\lab,v} \vfreach (\cxt', \allocation{o})\}
\end{equation}
where the reachability relation $\vfreach$ is now also
context-sensitive.

Passing parameters to and returning results from a callee
invoked at a callsite are handled by \rulename{C-CALL} and 
\rulename{C-Ret}. 
\rulename{C-Call} deals with the direct and
indirect value-flows backwards from the entry instruction of a
callee function to each of its callsites
based on the call graph 
computed on the fly similarly as \rulename{Call} in
Figure~\ref{fig:ddafs}, except that \rulename{C-Call} 
is context-sensitive.
Likewise, \rulename{C-Ret} deals with the direct and
indirect value-flows backwards from a callsite to
the return instruction of every callee function.

With context-sensitivity, \vfsu will filter out more spurious 
value-flows generated by Andersen's analysis, thereby producing more precise points-to information 
to enable more strong updates (\rulename{C-SU/WU}). At a store
$\cxt, \ell: *p = \_$, its kill set is context-sensitive. A strong
update is applied if $p$ points to a \emph{context-sensitive singleton}
$(\cxt',o')\in\cxtsingletons$, where  $o'$ is a (non-heap) singleton defined in 
Section~\ref{sec:dda} or a heap object with $\cxt'$ being
a \emph{concrete} context, i.e., one 
not involved in recursion or loops.

\begin{example}
\label{ex:cxt}
Let us use an example given
in Figure~\ref{fig:cxtex} to illustrate the effects of 
context-sensitive strong updates on computing the
points-to information for $z$ at $\lab_5$. 
This example is adapted from a real application,
\texttt{milc-v6}, given in Figure~\ref{fig:code_snippet}(c). 
Without context-sensitivity, \vfsu will only perform a 
weak update at $\lab_8: *x = y$, since $x$ 
points to both $a$ and $b$ passed into \texttt{foo()}
from the two callsites at $\lab_3$ and $\lab_4$. As a result,
$z$ at $\ell_5$ is found to point to not only
what $y$ points to, i.e., $c$ but also what $b$ points to
previously (not shown to avoid cluttering). With 
context-sensitivity, \vfsu finds that
$ \lvar{[\ ],\lab_{5},z}\vfreach
\lvar{[\ ],\lab_{5},b}\vfreach
\lvar{[\ ],\lab_4,b}\vfreach\lvar{[\lab_{4}],\lab_{9},b}\vfreach\lvar{[\lab_{4}],\lab_{8},b}\vfreach\lvar{[\lab_{4}],\lab_{7},y}\vfreach([\lab_{4}],\allocation{c})$. 
Since
$\lvar{[\lab_4], \lab_8, x}$ points to a context-sensitive
singleton $(\lab_4, b)$ at $\lab_8$, a strong update is performed to
$b$ at $\lab_8$, causing the old contents in $b$ to be killed.
\hfill {\small $\Box$}
\end{example}

Given a program, the
SCCs (strongly connected components) in its call graph are 
constructed on the fly. \vfsu handles the SCCs  in the
program
context-sensitively but the function calls inside a SCC 
context-insensitively as in \cite{Sridharan:2006}.

\begin{figure}[t]
\centering
\vspace*{-1cm}
\begin{tabular}{l}
\hfil\Large\textbf{\textcolor{blue}{}}\hfil\includegraphics[scale=0.65]{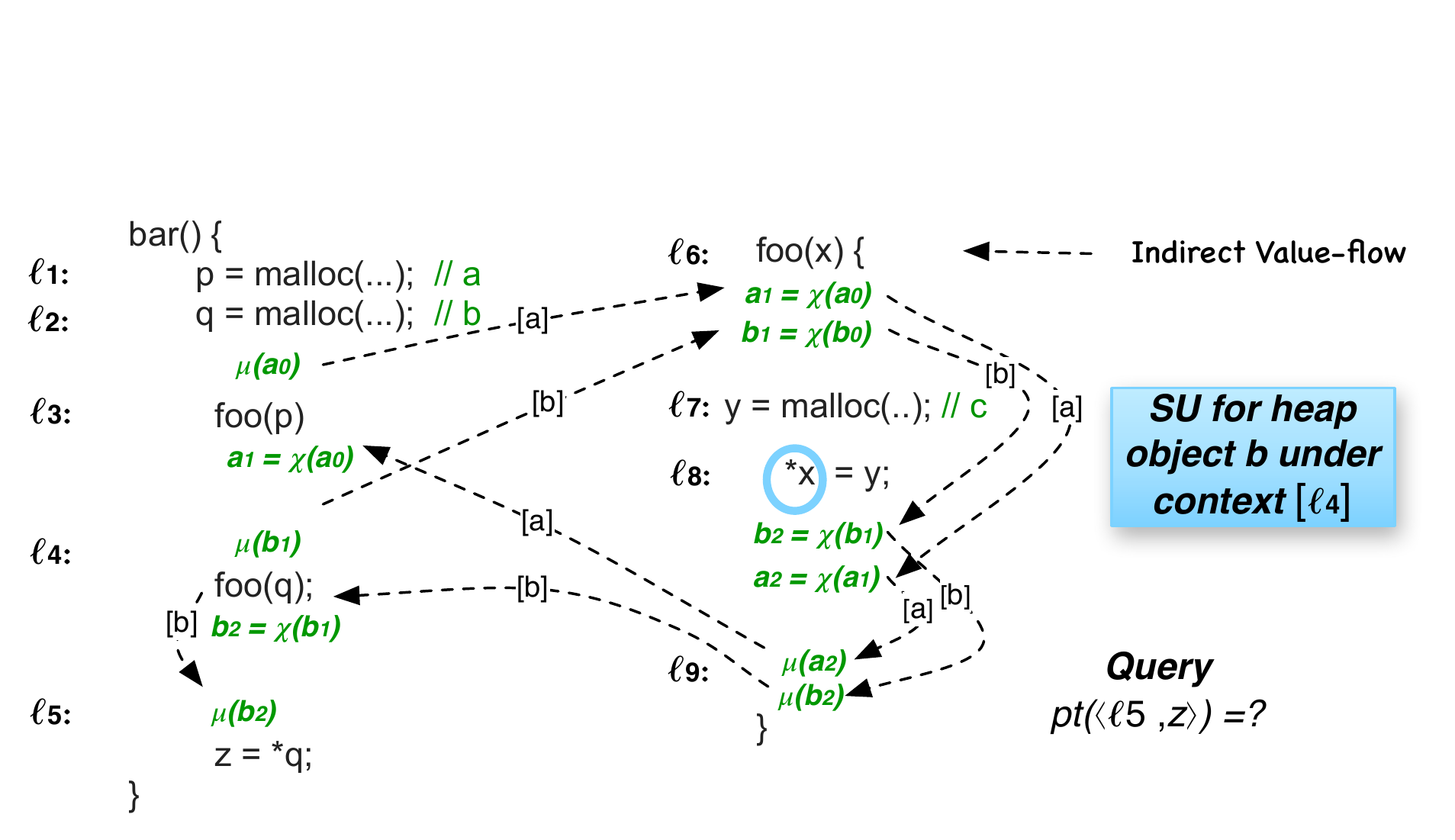}
\\
\hline
\end{tabular} 
\caption{
Resolving $\pts([\ ],\lab_5,z) = \{[\lab_4], c\}$) with 
	context-sensitive strong updates. 
\label{fig:cxtex} }
\end{figure}

\subsection{SUPA: Hybrid Multi-Stage Analysis
\label{sec:sensivity}}

To facilitate efficiency and precision tradeoffs in answering
on-demand queries, \vfsu, as illustrated in Figure~\ref{fig:framework},
organizes its analyses in multiple stages sorted in 
increasing efficiency but decreasing precision. 
Let there be $M$
queries issued successively. Let the $N$ stages of 
\vfsu,
$\textbf{\textit{Stage[0]}},\cdots, \textbf{\textit{Stage[N-1]}}$,
be configured with budgets $\bgt_0,\cdots, \bgt_{N-1}$, 
respectively.
In our current implementation, each budget is specified as the maximum
number of def-use chains traversed in the SVFG of the program.

\vfsu answers a query on-demand by applying its $N$ analyses 
successively, starting from $\textbf{\textit{Stage[0]}}$. If the query
is not answered after budget $\bgt_i$ has been exhausted at stage $i$,
\vfsu re-issues the query at stage $i+1$, and eventually
falls back to the results that are 
pre-computed by pre-analysis.

\vfsu caches fully computed points-to 
information in a query and reuses it in subsequent
queries, as illustrated in Figure~\ref{fig:mottrace}.  
Let $\calQ$ be the set of queried variables
issued from a program. Let $\calI\supseteq\calQ$ be the set 
of variables reached from $\calQ$ during the
analysis. Let $(\lab,v)\in \calQ$ be
a queried variable. We write
$\pts_{\bgt_i}^i(\lvar{\Delta_i,\lab,v})$ to
represent the points-to set of a variable $\lvar{\lab,v}$
computed at stage $i$ under budget $\bgt_i$, where
$\Delta_i$ is a contextual qualifier at stage $i$
(e.g., $\cxt$ in FSCS).  By convention,
$\pts_{\bgt_N}^N(\lvar{\Delta_N,\lab,v})$ denotes the
points-to set obtained by pre-analysis, at 
$\textbf{\textit{Stage[N]}}$ (conceptually). 

When resolving $\pts_{\bgt_i}^i(\lvar{\Delta_i,\lab,v})$
at stage $i$, suppose \vfsu has reached
a variable $\lvar{\lab',v'} \in\calI$ and needs to compute
$\pts_{*}^i(\lvar{\Delta_i,\lab',v'})$, where
$* (\leqslant \bgt_i$)
represents an unknown budget remaining,
with $(\lab',v')$ being $(\lab,v)$ possibly
(in a cycle).

Presently, \vfsu exploits two types of reuse
to improve efficiency with no loss of precision in a 
hybrid manner:
\begin{description}
\item[\fbox{Backward Reuse: $(\lab',v')\in \calI$}] If 
$\pts_{*}^j(\lvar{\Delta_j,\lab',v'})$, where $j\leqslant i$, was previously
cached, then
$\pts_{*}^i(\lvar{\Delta_i,\lab',v'}) = \pts_{*}^j(\lvar{\Delta_j,\lab',v'})$, provided that 
$\pts_{*}^j(\lvar{\Delta_j,\lab',v'})$ is a sound
representation of $\pts_{*}^i(\lvar{\Delta_i,\lab',v'})$.
For example, if 
$\textbf{\textit{Stage[i]}}=FS$ and $\textbf{\textit{Stage[j]}}=FSCS$,
then $\pts_{*}^{FSCS}(\lvar{\cxt',\lab',v'})$ can be reused for
$\pts_{*}^{FS}(\lvar{\lab',v'})$ if $\cxt'$ is \emph{true},
representing a context-free points-to set. 
\item[\fbox{Forward Reuse: $(\lab',v')\in \calQ$}]
If $\pts_{\bgt_j}^j(\lvar{\Delta_j,\lab',v'})$,
where $j>i$, was previously
computed and cached but
$\pts_{\bgt_k}^k(\lvar{\Delta_k,\lab',v'})$ was not,
where $0\leqslant k <j$, then \vfsu will also fail for
$\pts_{*}^k(\lvar{\Delta_k,\lab',v'})$, 
where $i \leqslant k<j$,
since $*\leqslant \bgt_k$. Therefore,
\vfsu will exploit the second type of reuse by setting
$\pts_{*}^i(\lvar{\Delta_i,\lab',v'}) = 
\pts_{\bgt_j}^j(\lvar{\Delta_j,\lab',v'})$.
\end{description}

Of course, many other schemes are possible with or without precision
loss.

\newcommand{\ato}{\textsf{UAO}\xspace}
\newcommand{\upc}{\textsf{Uninit}\xspace}

\section{EVALUATION}
\label{sec:eval}

We evaluate \vfsu by choosing detection of uninitialized pointers
as a major client. The objective
is to show that \vfsu is effective in answering client
queries, in environments with small time and memory budgets such as IDEs,
by facilitating efficiency and precision tradeoffs 
in a hybrid multi-stage analysis framework. 
We provide evidence to
demonstrate a good correlation between the number
of strong updates performed on-demand and the degree of precision
achieved during the analysis.

\subsection{Implementation}
\label{sec:impl}

We have implemented \vfsu 
in LLVM (3.5.0).  The source files of a program
are compiled under ``-O0''  (to facilitate 
detection of undefined values \cite{Zhao2012})
into bit-code by \texttt{clang} 
and then merged using the \texttt{LLVM Gold Plugin} 
at link time to produce a whole program bc file. The 
compiler option \texttt{mem2reg} is applied to
promote memory into registers. Otherwise,
SUPA will perform more strong updates on memory locations
that would otherwise be promoted to registers, favoring SUPA undesirably.

All the analyses evaluated are field-sensitive. 

Positive weight cycles 
that arise from processing fields of struct objects are collapsed~\cite{pearce2007efficient}. 
Arrays are considered monolithic so that the elements in
an array are not distinguished. 
Distinct allocation sites (i.e., \textsc{AddrOf} 
statements) are modeled by
distinct abstract objects.

We build the SVFG for a program based on our open-source software, \svf \cite{SVF}. The def-use chains are
pre-computed by 
Andersen's algorithm flow and context-insensitively.
In order to compute soundly and precisely the points-to information in
a value-flow cycle,
\vfsu retraverses the cycle whenever new points-to information is
discovered until a fix point is reached.

To compare \vfsu with whole-program analysis, 
we have implemented a sparse flow-sensitive (\SFS)
analysis described in~\cite{hardekopfflow}
also in LLVM, 
as \SFS is a recent solution yielding exactly
the flow-sensitive precision with good scalability.
However, there are some differences.
In \cite{hardekopfflow}, \SFS was implemented in 
LLVM (2.5.0), by using imprecisely 
pre-computed call graphs and representing
points-to sets with binary decision diagrams (BDDs).
In this paper, just like \vfsu, \SFS is 
implemented in LLVM (3.5.0), by building 
a program's call graph on the fly 
(Section~\ref{sec:dda}) and representing 
points-to sets with sparse bit vectors.

We have not implemented a whole-program 
FSCS pointer analysis in LLVM. There is no
open-source implementation either in
LLVM. According to~\cite{AchRob11}, existing FSCS 
algorithms for C ``do not scale even for an order of
magnitude smaller size programs than those analyzed'' 
by Andersen's algorithm. As shown here, \SFS can 
already spend hours on analyzing some programs 
under 500 KLOC.

\subsection{Methodology}
\label{sec:meth}

We choose uninitialized pointer detection as a
major client,
named \upc, which
requires strong update analysis to 
be effective.
As a common type of bugs in C programs, uninitialized
pointers are dangerous, as 
dereferencing them can cause system crashes and 
security vulnerabilities. For \upc,
flow-sensitivity is crucial. Otherwise,
strong updates are impossible, making \upc checks
futile.

We will show that
\vfsu can answer \upc's on-demand
queries efficiently while achieving nearly 
the same precision as \SFS.
For C, global and static variables are default initialized, but local variables are not. 
In order to mimic the default uninitialization
at a stack or heap allocation site
$\lab\!:\!p\!=\!\&a$ for an uninitialized pointer $a$, 
we add a special store 
$*p\!=\!u$ immediately after $\lab$,
where $u$ points to an \emph{unknown abstract object}
(\ato), $u_{a}$.
Given a load $x\!=\!*y$, we 
can issue a points-to query for $x$ to detect its
potential uninitialization.
If $x$ points to a $u_a$ (for some $a$), then
$x$ may be uninitialized. By performing 
strong updates more often, a flow-sensitive analysis can find
more \ato's that do not reach any pointer and  thus
prove more pointers to be initialized. Note that \SFS 
can yield false positives since, for example, path
correlations are not modeled.

We do not introduce $\ato$'s
for the local variables involved in recursion and
array objects since they cannot be 
strongly updated. We also ignore all the 
default-initialized stack or heap 
objects (e.g., those created by \texttt{calloc()}).

\begin{table}[t]
	\centering
		\caption{Program characteristics.
		\label{tab:stat}}
	\scalebox{1}{
\addtolength{\tabcolsep}{0pt}
		\begin{tabular}{|l | |l | l | l | l | l | l|}
			\hline
			\multicolumn{1}{|c||}{Program}	&KLOC	&Statements	&\multicolumn{1}{c|}{Pointers}	&\multicolumn{1}{c|}{Allocation Sites} &\multicolumn{1}{c|}{Queries} \\ \hline
			\hline
			spell-1.1   &0.8 & 1011      &1274      &42   & 17 \\ \hline
			bc-1.06   &14.4 &17018	&15212	&654	   &689 \\ \hline
			milc-v6	&15	&11713	&29584	&865	&3	\\ \hline
			less-451	&27.1	&6766	&22835	&1135	&100\\ \hline
			sed-4.2     &38.6       &25835     &34226    &395        & 1191 \\ \hline
			hmmer-2.3	&36	&27924	&74689	&1472	&2043\\ \hline
			make-4.1	&40.4	&14926	&36707	&1563	&1133\\ \hline
			gzip-1.6          &64.4        &22028    &25646     &1180      &551 \\ \hline
			a2ps-4.14	&64.6	&49172	&116129	&3625	&5065\\ \hline
			bison-3.0.4	&113.3	&36815	&90049	&1976	&4408\\ \hline
			grep-2.21	&118.4	&10199	&33931	&1108	&562\\ \hline
			tar-1.28	&132	&30504	&85727	&3350	&909\\ \hline
			wget-1.16  &140.0 &51556 & 63199  &726   	&1142\\ \hline
			bash-4.3	&155.9	&59442	&191413	&6359	&5103\\ \hline
			gnugo-3.4   &197.2 & 369741      &286986      &27511   & 1970 \\ \hline
			sendmail-8.15	&259.9	&86653	&256074	&7549	&2715\\ \hline
			vim-7.4	&413.1	&147550	&466493	&8960	&6753\\ \hline
			emacs-24.4	&431.9	&189097	&754746	&12037	&4438\\ \hline
Total&	2263.0&	1157950&	2584920&	80507&	38792 \\ \hline
		\end{tabular}
	}
	\end{table}

We generate meaningful points-to queries, one query 
for the top-level variable $x$ at each load $x\!=\!*y$. However, we ignore this
query if $x$ is found not to point to any $\ato$ 
by pre-analysis. 
This happens only when 
$x$ points to either default-initialized
objects or unmodeled
local variables in recursion cycles or arrays.
The number of queries issued in each program is
listed in the last column in Table~\ref{tab:stat}.

\subsection{Experimental Setup}

We use a machine
with a 3.7GHz Intel Xeon 8-core CPU and 64 GB memory.
As shown in Table~\ref{tab:stat}, we have selected a total of
18 open-source programs 
from a variety of domains:
\texttt{spell-1.1} (a spelling checker),
\texttt{bc-1.06} (a numeric processing language),
\texttt{milc-v6} (quantum chromodynamics), 
\texttt{less-451} (a terminal pager),  
\texttt{sed-4.2} (a stream editor),  
\texttt{milc-v6} (quantum chromodynamics), 
\texttt{hmmer-2.3} (sequence similarity searching), 
\texttt{make-4.1} (a build automation tool), 
\texttt{a2ps-4.14} (a postScript filter), 
\texttt{bison-3.04} (a parser), 
\texttt{grep-2.2.1} (string searching),
\texttt{tar-1.28} (tar archiving),
\texttt{wget-1.16} (a file downloading tool),
\texttt{bash-4.3} (a unix shell and command language),
\texttt{gnugo-3.4} (a Go game),
\texttt{sendmail-8.15.1} (an email server and client),
\texttt{vim74} (a text editor), and 
\texttt{emacs-24.4} (a text editor).

For each program, Table~\ref{tab:stat} lists its
number of lines of code, statements which are LLVM instructions relevant to our pointer analysis, 
pointers, allocation sites (or \texttt{AddrOf} statements), and 
queries issued (as discussed in Section~\ref{sec:meth}).

\subsection{Results and Analysis}

We evaluate \vfsu with two configurations,
\vfsuFS and $\vfsuFSCS$.  
\vfsuFS is a one-stage FS analysis by 
considering 
flow-sensitivity only. $\vfsuFSCS$ is a two-stage analysis
consisting of FSCS and FS applied in that order.

\subsubsection{Evaluating \vfsuFS}

When assessing $\vfsuFS$, we consider two 
different criteria:
efficiency (its analysis time and memory 
usage per query) and precision (its competitiveness 
against \SFS). 
For each query, its analysis budget, denoted $B$, 
represents the maximum number of def-use chains that can be
traversed. We consider a wide range of budgets with $B$ falling
into $[10, 200000]$.

\vfsuFS is highly effectively. With $B=10000$,
\vfsuFS is nearly as precise as \SFS,
by consuming about 0.18 seconds
and 65KB of memory per query, on average.

\begin{figure*}[th]
\begin{center}
\scalebox{0.73}{
\begin{tabular}{l@{\hspace{10mm}}l}

\hfil\Large\textbf{\textcolor{blue}{}}\hfil\includegraphics[scale=0.84]{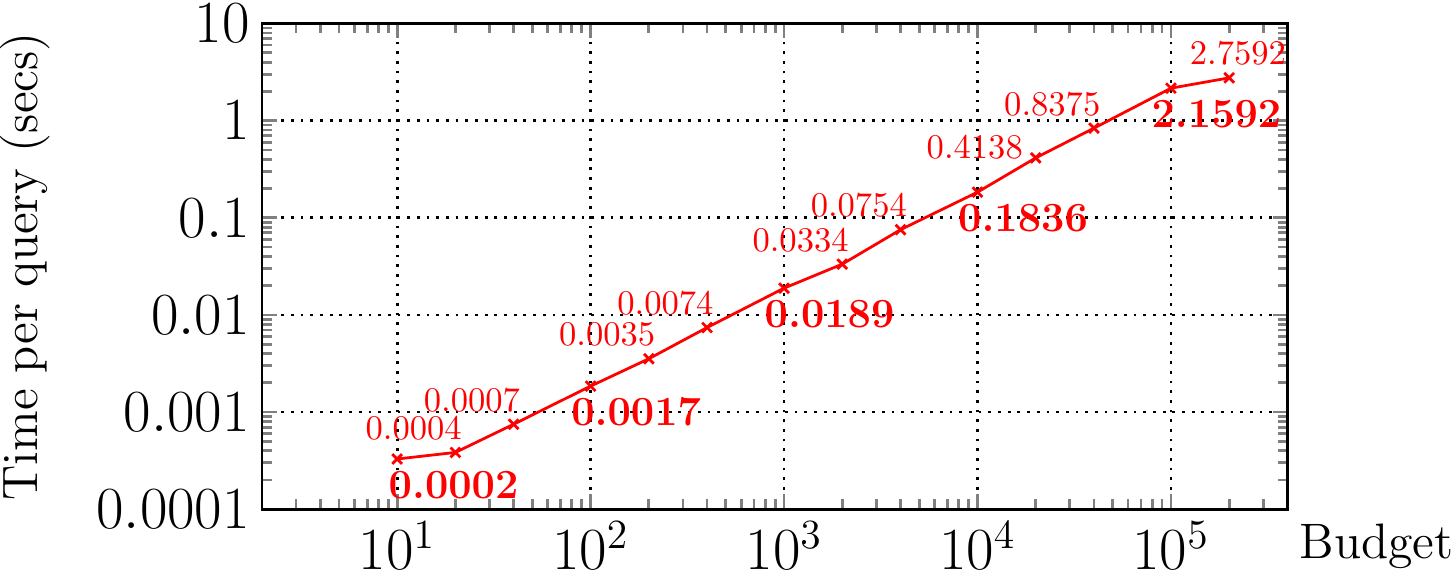}
\\
\hspace{45mm} \Large (a)  Analysis Time\\ \\

\hfil\Large\textbf{\textcolor{blue}{}}\hfil\includegraphics[scale=0.84]{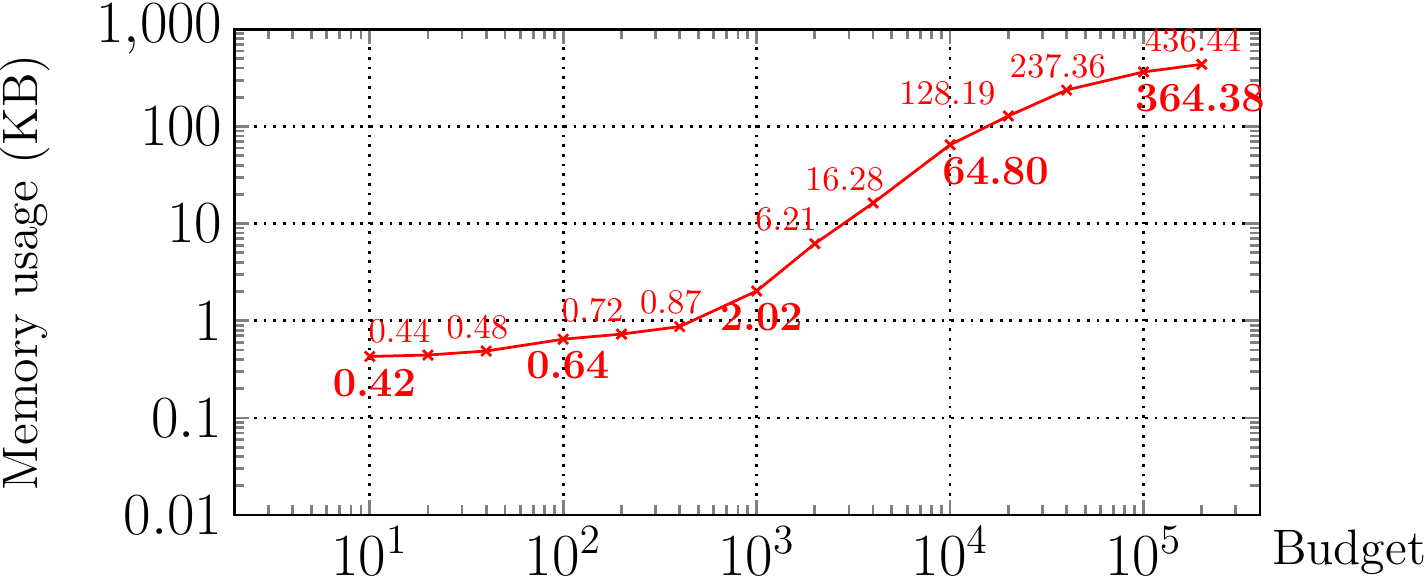}
\\[0.5ex]
\hspace{45mm}\Large (b) Memory Usage
\end{tabular}
}
\end{center}
\caption{Average analysis time and memory usage
per query consumed by \vfsuFS under different analysis 
budgets (with both axes being logarithmic).}
\label{fig:running_time}

\end{figure*}

\begin{table}[hbt]
        \begin{center}
\caption{Pre-processing times taken by pre-analysis shared by
        \vfsu and \SFS and analysis times of \SFS (in seconds). \label{fig:pre-analysis}}
\scalebox{1}{
\renewcommand{\arraystretch}{1}
\addtolength{\tabcolsep}{0pt}
                \begin{tabular}{|c||r|r|r|r|r|}
\hline
\multirow{3}{*}{Program} & \multicolumn{3}{c|}{Pre-Analysis Times } & \multicolumn{1}{c|}{\multirow{3}{*}{Analysis Time of SFS}} \\ 
& \multicolumn{3}{c|}{Shared by \vfsu and \SFS} & \\ \cline{2-4}
& \multicolumn{1}{c|}{Andersen's Analysis}  & \multicolumn{1}{c|}{SVFG} & \multicolumn{1}{c|}{Total} &  \\ \hline \hline
spell  &0.01	&0.01&0.01	&0.01 \\\hline
bc     &0.35	&0.21&0.56	&0.98 \\\hline	
milc	&0.42	&0.1	&0.52	&0.16 \\\hline
less	&0.42	&0.37	&0.79	&1.94	\\\hline
sed	&1.38	&0.34	&1.73	&5.46	\\\hline
hmmer	&1.57	&0.46	&2.03	&1.07	\\\hline
make	&1.74	&1.17	&2.91	&13.94	\\\hline
gzip	&0.27	&0.10	&0.37	&0.20	\\\hline
a2ps	&7.34	&1.31	&8.65	&60.61	\\\hline
bison	&8.18	&3.66	&11.84	&44.16	\\\hline
grep	&1.44	&0.17	&1.61	&2.39	\\\hline
tar	&2.73	&1.71	&4.44	&12.27	\\\hline
wget	&1.86	&0.90	&2.76	&3.47	\\\hline
bash	&53.48	&44.07	&97.55	&2590.69	\\\hline
gnugo  &5.68	&2.75&8.44	&9.86 \\\hline
sendmail	&24.05	&23.43	&47.48	&348.63\\\hline
vim	&445.88	&85.69	&531.57	&13823\\\hline
emacs	&135.93	&146.94	&282.87	&8047.55\\\hline
                 \end{tabular}
}
         \end{center}
\end{table}

\paragraph{\underline{Efficiency}} 
Figure~\ref{fig:running_time}(a) shows 
the average analysis time per query 
for all the programs
under a given budget, with about
0.18 seconds when $B\!=\!10000$ and about
2.76 seconds when
$B\!=\!200000$. 
Both axes are logarithmic. 
The longest-running queries can take an order 
of magnitude as long as the average cases. However, 
most queries 
(around 80\% across the programs) 
take much less than the average cases. 
Take \textsf{emacs} for example. 
\SFS takes over two hours (8047.55 seconds)
to finish. In contrast, \vfsuFS spends less than 
ten minutes (502.10 seconds) 
when $B\!=\!2000$, with an average per-query
time (memory usage) of 
0.18 seconds (0.12KB), and produces 
the same answers for all the queries as \SFS
(shown in Figure~\ref{fig:precision} and explained
below).

For \vfsu, its pre-analysis is lightweight,
as shown in Table~\ref{fig:pre-analysis},
with \texttt{vim} taking the longest at 531.57
seconds. The same pre-analysis
is also shared by \SFS in order to enable its own
sparse whole-program
analysis.
The additional time taken by \SFS for analyzing
each program entirely is given in the last column.


Figure~\ref{fig:running_time}(b) shows the average
memory 
usage per query under different budgets. Following the common practice, we measure the real-time memory usage by reading the virtual memory information (\texttt{VmSize}) from the
linux kernel file (\texttt{/proc/self/status}).
The memory monitor starts after the pre-analysis 
to measure the memory usage for answering queries 
only. The average amount of memory consumed per 
query is small, with about 65KB when $B=10000$
and about 436KB when $B\!=\!200000$. 
Even under the largest budget $B=200000$ evaluated, 
\vfsuFS never uses more than 3MB for any single query
processed.

\begin{figure}[th]
\begin{center}
\scalebox{0.74}{
\begin{tabular}{l}
\hspace{-7mm}
\hfil\Large\textbf{\textcolor{blue}{}}\hfil\includegraphics[scale=1.02]{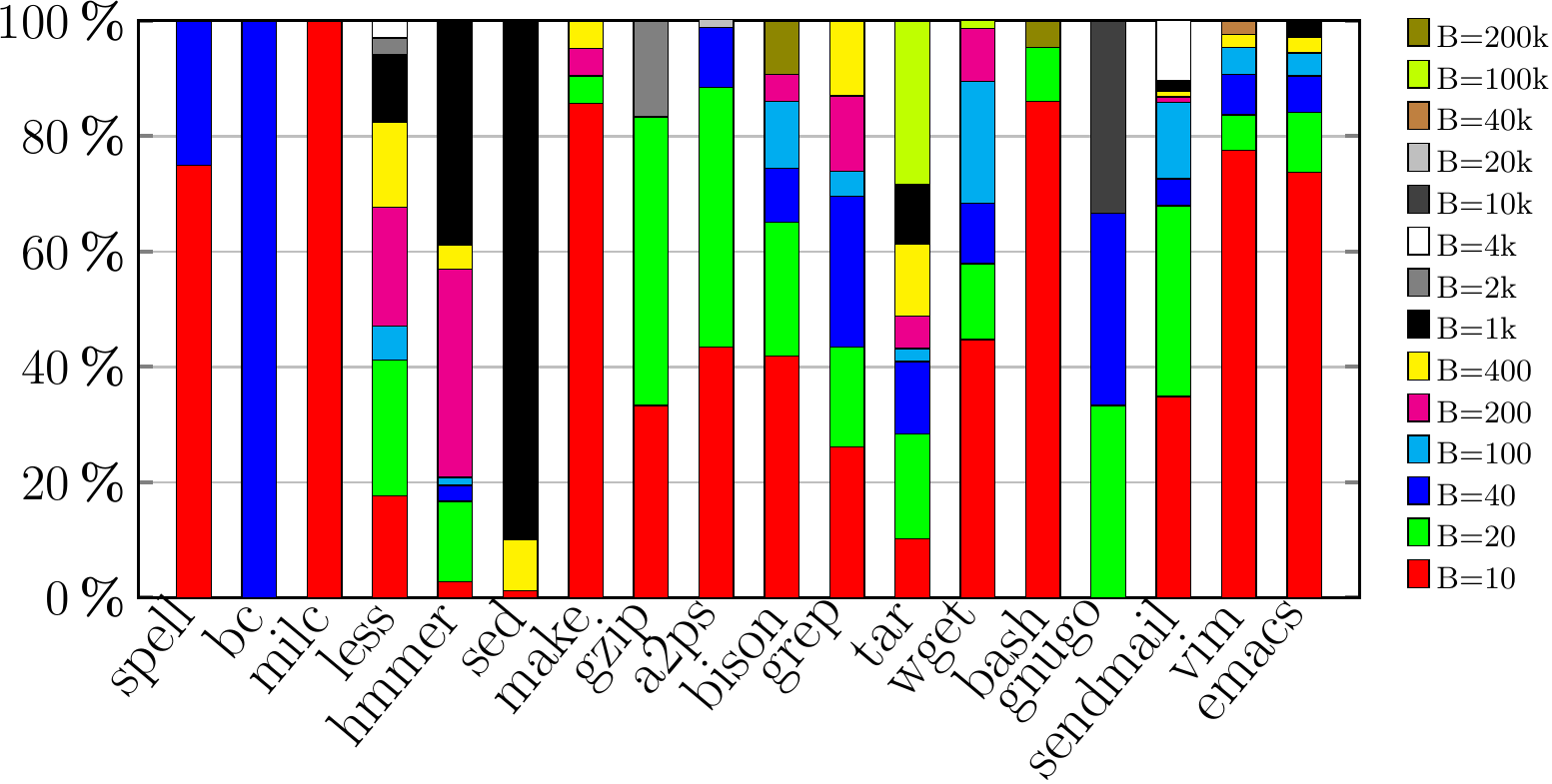}
\\ \hline
\end{tabular} 
}
\end{center}
\caption{Percentage of queried variables
proved to be initialized by \vfsuFS over \SFS
under different budgets.
\label{fig:precision}
}
\end{figure}

\paragraph*{\underline{Precision}}
Given a query $\pts({\lvar{\lab,p}}$), 
$p$ is initialized if no \ato is pointed
by $p$ and potentially uninitialized otherwise. 
We measure the precision of \vfsuFS in terms of 
the percentage of queried variables
proved to be initialized by comparing with \SFS,
which yields the best precision achievable
as a whole-program flow-sensitive analysis.

Figure~\ref{fig:precision} reports our results.
As $B$ increases, the precision of \vfsuFS 
generally improves. 
With $B\!=\!10000$, \vfsuFS can answer correctly
97.4\% of all the queries from the 18 programs.
These results indicate that our analysis is 
highly accurate, even under tight budgets.
For the 18 programs except
\texttt{a2ps}, \texttt{bison} and \texttt{bash}, 
\vfsuFS produces the same answers for all the
queries when $B=100000$ as \SFS. When
$B=200000$ for these three programs, 
\vfsu becomes as precise as \SFS, by taking an average
of 0.02 seconds (88.5KB) for 
\texttt{a2ps},
0.25 seconds (194.7KB)
for \texttt{bison}, and 
3.18 seconds (1139.3KB) for \texttt{bash}, per query. 

\begin{figure*}[h]
  \begin{center}
\begin{tabular}{l}
\hspace{-3mm}
\hfil\Large\textbf{\textcolor{blue}{\hfil\includegraphics[scale=0.73]{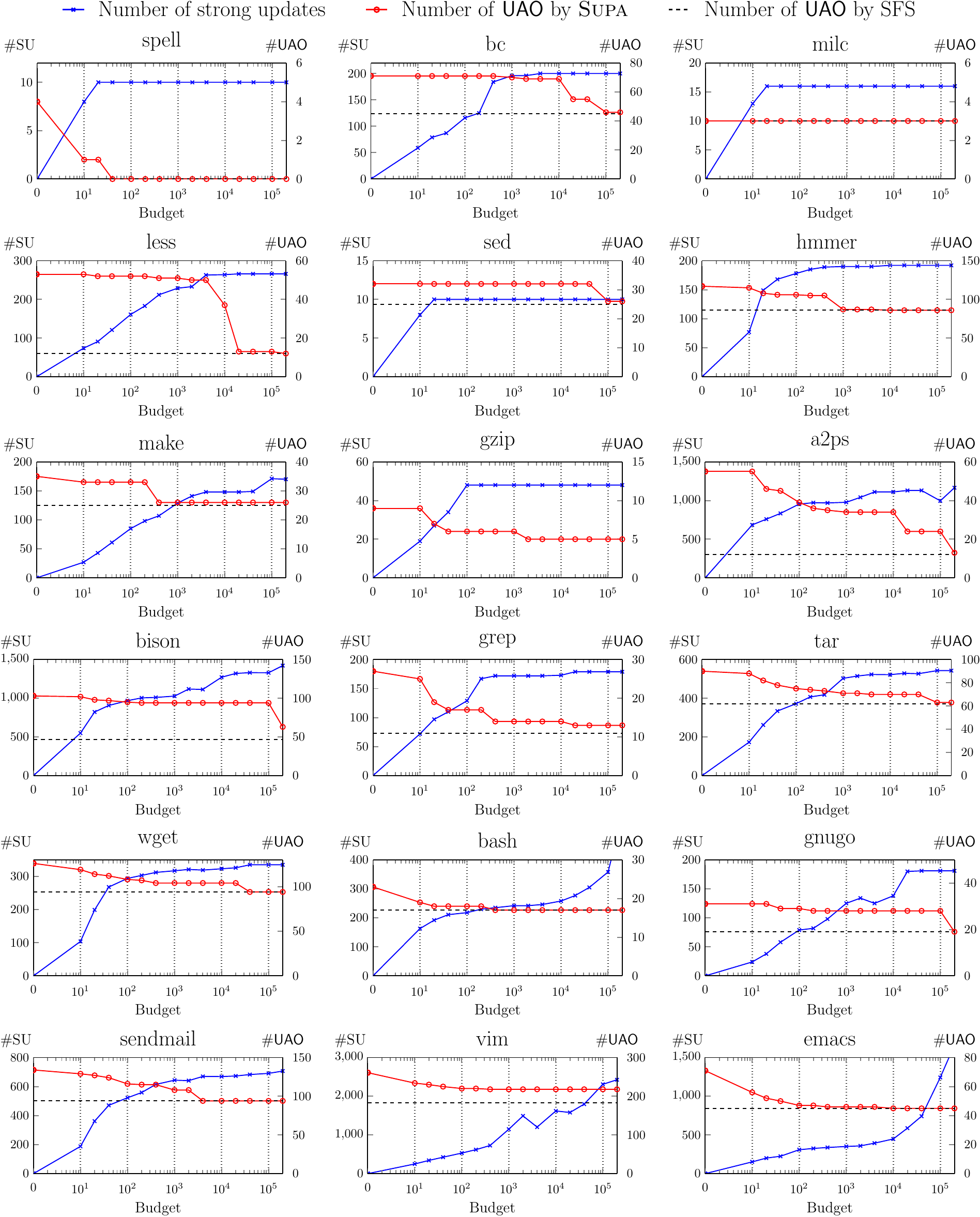}}}
\\ \hline
\end{tabular} 
  \end{center}
  \caption{Correlating the number of strong updates 
  with the number of \ato's under \vfsuFS with 
  different analysis budgets.}
  \label{fig:obj_su}
\end{figure*}

\paragraph*{\underline{Understanding On-Demand Strong Updates}} 
Let us examine the benefits achieved by 
\vfsuFS in answering client queries by
applying on-demand strong updates.
For each program, Figure~\ref{fig:obj_su}
shows a good correlation between the number of 
strong updates performed (\#SU on the left
y-axis) in a blue curve
and the number of \ato's reaching some 
uninitialized pointers (\#\ato on the right
y-axis) in a red curve under varying budgets (on the
logarithmic x-axis). The number
of such \ato's reported by \SFS is shown as
the lower bound for \vfsuFS in a dashed line.

In most programs,
\vfsuFS performs increasingly more
strong updates to block increasingly more \ato's to
reach the queried variables as the
analysis budget $B$ increases,  because \vfsuFS falls back 
increasingly less frequently from FS to the
pre-computed points-to information. 
When $B$ increases, \vfsuFS
can filter out more spurious value-flows in the SVFG 
to obtain more precise points-to information, thereby 
enabling more strong updates to kill the \ato's.

When $B=200000$, \vfsuFS gives the same answers as
\SFS in all the 18 programs except
\texttt{bison} and
\texttt{vim}, which causes 
\vfsuFS to report 16 and 35 more, respectively.

For some programs such as \texttt{spell}, \texttt{bc}, \texttt{milc},
\texttt{hmmer} and \texttt{grep}, most of their strong 
updates happen under small budgets (e.g.,
$B=1000$).  In \texttt{hmmer}, for example,
192 strong updates are performed when $B=10000$. 
Of the 5126 queries issued, \vfsuFS runs 
out-of-budget for only three queries, which
are all fully resolved when
$B=200000$, but with no further 
strong updates being observed.

For programs like \texttt{bison}, \texttt{bash}, \texttt{gnugo} and 
\texttt{emacs},
quite a few strong updates take place when $B>1000$. 
There are 
two main reasons. First, these programs
have many indirect call edges (with
8709 in \texttt{bison},
1286 in \texttt{bash},
23150 in \texttt{gnugo} and
4708 in \texttt{emacs}), 
making their on-the-fly call graph construction 
costly (Section~4.1.2). Second, there are many
value-flow cycles (with over 50\% def-use chains 
occurring in cycles in \textsf{bison}), making their
constraint resolution
costly (to reach a fixed point). Therefore,
relatively large budgets are needed to enable more strong updates to be performed.

Interestingly, in programs such as \texttt{a2ps}, \texttt{gnugo} and 
\texttt{vim}, fewer strong updates are observed
when larger budgets are used. 
In \texttt{vim}, the number of strong updates
performed is 1492 when $B=2000$ but drops to
1204 when $B=4000$. This is due to the forward
reuse described in Section~\ref{sec:sensivity}.
When answering a query $\pts(\lvar{\lab,v})$
under two budgets $B_1$ and $B_2$,
where $B_1<B_2$, \vfsuFS has reached 
$\lvar{\lab',v'}$ and needs to compute
$\pts(\lvar{\lab',v'})$ in each case.
\vfsuFS may fall back to the flow-insensitive 
points-to set of
$v'$ under $B_1$ but not $B_2$, 
resulting
in more strong updates performed under $B_1$ in the 
part of the program that is not explored under $B_2$.

\subsubsection{Evaluating \vfsuFSCS}

For C programs, flow-sensitivity is regarded 
as being important for achieving useful high precision. 
However, context-sensitivity can be important
for some C programs, in terms of both obtaining 
more precise points-to information and enabling more
strong updates. Unfortunately, whole-program
analysis does not scale well to large programs when both
are considered (Section~\ref{sec:impl}).

\begin{table}[th]
\centering
\caption{Average analysis times consumed and \ato's
reported by \vfsuFSCS (with a budget of 10000 in each
stage) and \vfsuFS (with a budget of 10000 in total).
\label{fig:cxt}}
\renewcommand{\arraystretch}{1}
\addtolength{\tabcolsep}{1pt}
\begin{tabular}{|l|r|r|r|r|}
\hline
\multirow{2}{*}{Program} & \multicolumn{2}{c|}{\vfsuFS} & \multicolumn{2}{c|}{\vfsuFSCS} \\ \cline{2-5}
 & Time (ms) & $\#\ato$ & Time (ms) & $\#\ato$ \\ \hline \hline
spell & 0.01  & 0 & 0.01   & 0 \\ \hline
bc & 18.35  & 69 & 287.23   & 69 \\ \hline
 milc & 0.02  & 3 & 14.52  & \bf{0} \\ \hline
 less & 15.15  & 37 & 92.41  & 37 \\ \hline
 sed & 355.60  & 32 & 4725.42   & 32 \\ \hline
 hmmer & 11.41  & 86 & 135.05  & \bf{71} \\ \hline
 make & 124.40  & 26 & 229.44  & 26 \\ \hline
 gzip & 0.64  & 5 & 4.28  & 5 \\ \hline 
 a2ps & 126.01  & 34 & 448.26  & \bf{32} \\ \hline
 bison & 465.54  & 94 & 529.20  & \bf{86} \\ \hline
 grep & 124.46  & 14 & 197.66  & 14 \\ \hline
 tar & 26.31  & 70 & 83.10  & \bf{68} \\ \hline
 wget & 24.51  & 104 & 84.90  & {104} \\ \hline 
 bash & 188.69  & 17 & 327.16  & 17 \\ \hline
 gnugo & 72.73  & 28 & 80.08   & \bf{27} \\ \hline
 sendmail & 200.32  & 94 & 250.19  & \bf{85} \\ \hline
 vim & 168.67  & 218 & 473.25  & 218 \\ \hline
 emacs & 159.22  & 45 & 222.65  & 45 \\ \hline
 \end{tabular}
\end{table}

In this section, we demonstrate that \vfsu
can exploit both flow- and context-sensitivity
effectively \emph{on-demand} in a hybrid
multi-stage analysis framework, providing improved
precision needed by some programs.
Table~\ref{fig:cxt} compares \vfsuFSCS (with 
a budget of 20000 divided evenly in its
FSCS and FS stages) with \vfsuFS (with a budget 
of 10000 in its single FS stage). The maximal depth
of a context stack allowed is 3.
By allocating
the budgets this way, we can investigate some
additional precision benefits achieved by considering
both flow- and context-sensitivity.

In general, \vfsuFSCS has longer query
response times than \vfsuFS due to the larger
budgets used in our setting and the times taken
in handling context-sensitivity.
In \texttt{milc}, \texttt{hmmer}, \texttt{a2ps}, \texttt{bison}, \texttt{tar} , \texttt{gnugo} and \texttt{sendmail}, 
\vfsuFSCS reports fewer \ato's than \vfsuFS, for
two reasons. First, \vfsuFSCS
can perform strong updates context-sensitively
for stack and global objects, resulting in 
0 \ato's reported by \vfsuFSCS for \textsf{milc}.
Second, \vfsuFSCS can perform strong updates
to context-sensitive singleton heap objects defined in
Section~\ref{sec:inter}, by eliminating
8 \ato's in \texttt{bison}, 1 in \texttt{tar} 
and 1 in \texttt{sendmail}, which have been reported by \vfsuFS.

\section{Case Studies}
\label{sec:cases}

\begin{figure*}[ht]
\centering
\scalebox{0.8}{
\begin{tabular}{@{}l@{}}
\hspace*{-1ex}
\hfil\Large\textbf{\textcolor{blue}{}}\hfil\includegraphics[scale=1]{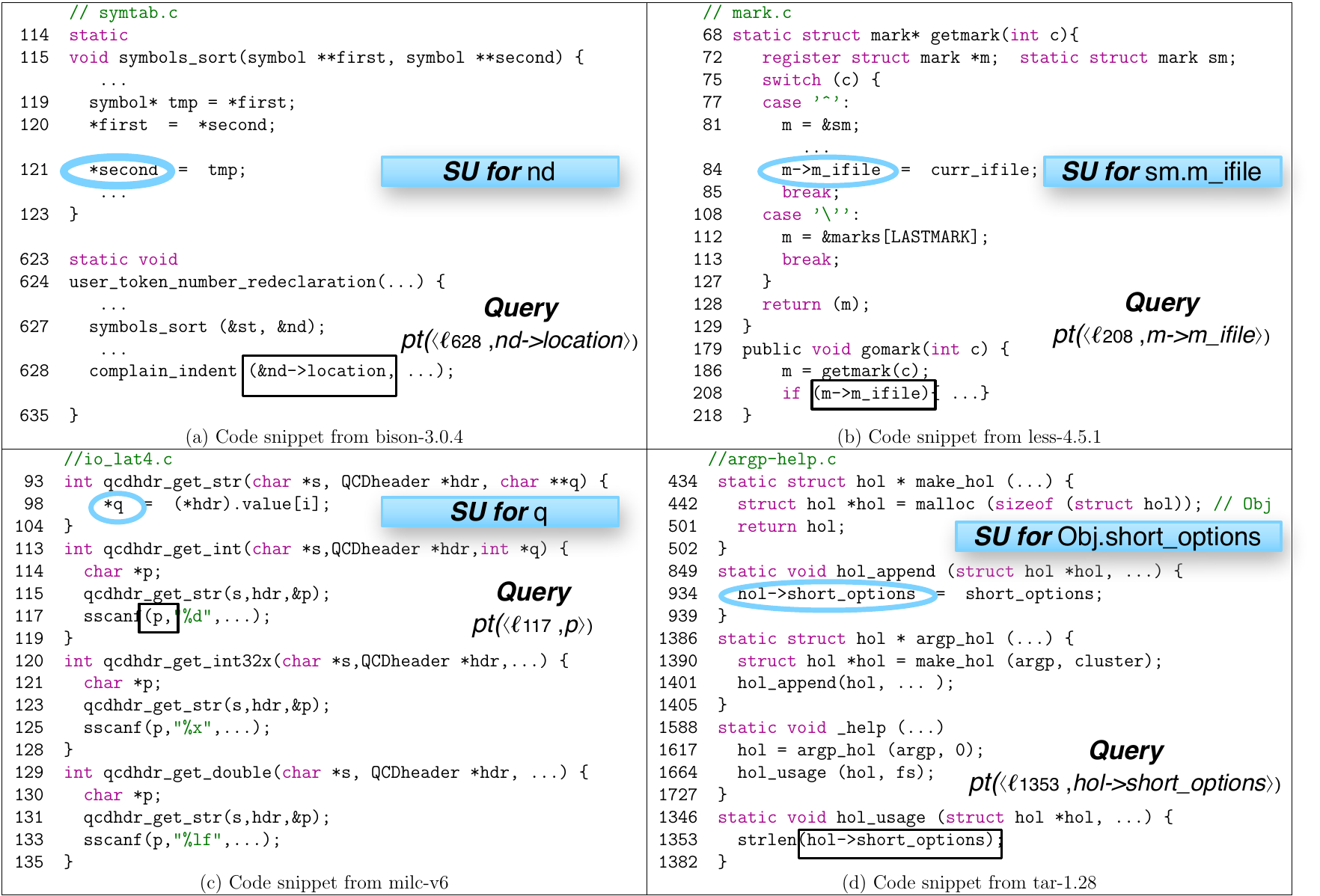}
\end{tabular} 
}
\caption{Selected code snippets.}
\label{fig:code_snippet}
\end{figure*}

We examine some real code to see how client 
queries are answered precisely with on-demand strong updates 
under four different scenarios.  
\begin{description}
\item[Figure~\ref{fig:code_snippet}(a)] There is 
a swap from \texttt{bison}. In line 121,
\texttt{second} points to a singleton stack object 
\texttt{nd} passed from line 627. So a strong
update is applied. When querying 
\texttt{nd->location} in line 628, \vfsu knows that
\texttt{nd} points to what
\texttt{st} pointed to before.

\item [Figure~\ref{fig:code_snippet}(b)]
In the code fragment from \texttt{less},
\texttt{m->m\_ifile} is initialized in two 
different branches, one recognized due to a 
strong update performed at the store  in
line 84 and one due to the default initialization
in line 112. According to \vfsu,
\texttt{m->m\_ifile} in line 208 is initialized.

\item[Figure~\ref{fig:code_snippet}(c)]

In the code fragment from \texttt{milc}, \texttt{q} in line 98 can 
point to several stack variables that are
all named $p$ in lines 115, 123 and 131.
With context-sensitivity, \vfsu finds that
\texttt{q} points to one 
singleton under each context. Thus,
a strong update is performed 
so that each stack variable becomes properly 
initialized when queried at each call to
\texttt{sscanf()}.

\item[Figure~\ref{fig:code_snippet}(d)]
In the code fragment from 
\texttt{tar}, \texttt{hol} in line 1390 points to 
a heap object $o$ allocated in line 442. With $o$
treated as a context-sensitive singleton (requiring
a context stack of at least depth 1), a strong update
can be performed in line 934 to initialize 
its field \texttt{short\_options} properly.  
\end{description}

\section{Parallelizing \vfsu}
\label{sec:para}

To demonstrate that \vfsu is amenable to 
parallelization as a demand-driven analysis,
we have parallelized \vfsuFS by using Intel 
Threading Building Blocks (TBB). 
A \emph{concurrent\_queue} is used 
to store all the queries issued from a program. 
We use a \emph{task\_group} to allocate tasks for 
computing the queries from \emph{concurrent\_queue}
in parallel. The cached points-to information
is shared with a \emph{concurrent\_hash\_map}.

\begin{figure*}[t]
\begin{tabular}{l}
\hspace{-8mm}
\hfil\Large\textbf{\textcolor{blue}{}}\hfil\includegraphics[scale=0.28]{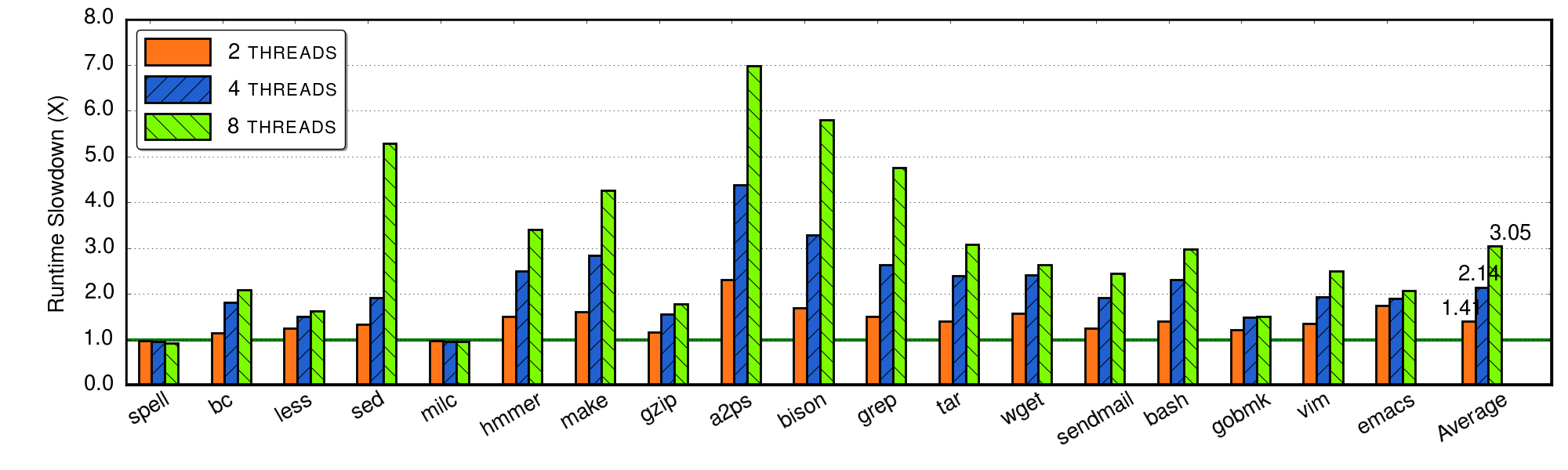}
\end{tabular} 
\caption{Speedups of \vfsuFS when parallelized
over its sequential version with two, four and eight threads
($B=10000)$.}
\label{fig:speedup}
\end{figure*}

Figure~\ref{fig:speedup} shows the speedups achieved
by parallelization over the sequential setting with
$B=10000$. 
With eight threads, the average speedup for the
18 programs is 3.05x and the
maximum speedup observed is 6.9x at \texttt{a2ps}. The time for each setting excludes the pre-analysis time.
Some programs enjoy better
speedups than others. There are three main reasons. First, 
some programs, such as \texttt{spell}, \texttt{less} and 
\texttt{milc}, have relatively few queries issued. 
Therefore, the performance benefits achieved from
query parallelization can be small. Second, different
queries take different times to answer, resulting in
different degrees of workload imbalance in different programs.
Third, different programs suffer from different
synchronization overheads in accessing the cached
points-to information in \emph{concurrent\_hash\_map}.

\section{RELATED WORK}
\label{sec:rela}

Demand-driven and whole-program approaches 
represent two important solutions to 
long-standing pointer analysis problems.
While a whole-program pointer analysis aims to
resolve all the pointers in the program, a 
demand-driven pointer analysis is designed to resolve
only a (typically small) subset of the set of these pointers in a 
client-specific manner.
This work is not concerned with developing  
an ultra-fast whole-program pointer analysis. Rather, 
our objective is to design a staged demand-driven strong 
update analysis framework that facilitates
efficiency and precision tradeoffs
flow- and context-sensitively according to the needs of a client
(e.g., user-specified budgets).
Below
we limit our discussion to the work that is 
most relevant to \vfsu.

\subsection{Flow-Sensitive Pointer Analysis}

Strong updates require pointers to be analyzed flow-sensitively with respect to program execution order. Whole-program flow-sensitive pointer analysis has been studied extensively in the literature.  
\citeauthor{choi1993efficient}~\citeyear{choi1993efficient} and \citeauthor{emami1994context}~\citeyear{emami1994context} gave some formulations in an iterative
data-flow framework~\cite{kam1977monotone}. 
\citeauthor{wilson1995efficient}~\citeyear{wilson1995efficient}  considered
both flow- and context-sensitivity by representing procedure
summaries with partial transfer functions, but restricted
strong updates to top-level variables only.
To eliminate unnecessary propagation of points-to information during the
iterative data-flow analysis~\cite{yu2010level,hardekopf2009semi,hardekopfflow,oh2012design}, some form of sparsity has been exploited.
The sparse value-flows, i.e., def-use chains in a program are captured by
sparse evaluation graphs (SEG)~\cite{choi1991automatic,ramalingam2002sparse}
as in \cite{hind1998assessing} and various
SSA representations such as HSSA~\cite{chow1996effective},
partial SSA~\cite{lattner2004llvm} and SSI~\cite{ananian1999static,tavares2014parameterized}. 
The def-use chains for top-level pointers,
once put in SSA, can be explicitly and precisely 
identified, giving rise to a so-called semi-sparse 
flow-sensitive analysis~\cite{hardekopf2009semi}.
Later, the idea of staged analysis
\cite{fink2008effective}
has been leveraged to make pointer analysis full-sparse for both top-level and address-taken variables by using fast Andersen's analysis as precise analysis \cite{sui2016sparse,ye2014region,hardekopfflow}. This paper is the first to exploit sparsity to improve the performance of a flow- and context-sensitive
demand-driven analysis with strong updates being performed
for C programs.

Recently, Balatsouras and Smaragdakis \cite{George2016} propose a fine-grained field-sensitive modeling technique for
performing Andersen's analysis by inferring lazily the types of heap objects in order to filter out redundant 
field derivations. This technique can be exploited to obtain
a more precise pre-analysis to improve the precision and/or 
efficiency of sparse flow-sensitive analysis.

\subsection{Demand-Driven Pointer Analysis}

Demand-driven pointer analyses
for C~\cite{Heintze:2001:DPA,Zheng:2008,Zhang:2014:ESA} 
and Java~\cite{Lu13,Shang12,yan2011demand,Sridharan:2006,Su15} are
flow-insensitive, formulated in terms of 
CFL (Context-Free-Language) reachability~\cite{Reps:1995:PID}.
\citeauthor{Heintze:2001:DPA}~\citeyear{Heintze:2001:DPA} introduced the first
on-demand Andersen-style pointer analysis for C. 
Later, \citeauthor{Zheng:2008}~\citeyear{Zheng:2008} performed alias analysis 
for C in terms of CFL-reachability flow- and context-insensitively 
with indirect function calls handled conservatively. 
Sridharan et al. gave two CFL-reachability-based formulations for 
Java, initially without considering 
context-sensitivity~\cite{Sridharan:2005:DPA} and later with 
context-sensitivity \cite{Sridharan:2006}.
\citeauthor{Shang12}~\citeyear{Shang12} and \citeauthor{yan2011demand}~ \citeyear{yan2011demand}
investigated how to summarize points-to information discovered
during the CFL-reachability analysis to improve 
performance for Java programs. \citeauthor{Lu13}~\citeyear{Lu13}
introduced an 
incremental pointer analysis with a CFL-reachability 
formulation for Java. \citeauthor{Su14}~\citeyear{Su14} demonstrated that
the CFL-reachability formulation is highly
amenable to parallelization on multi-core CPUs.
Recently, \citeauthor{Feng15}~\citeyear{Feng15} focused on answering demand
queries for Java programs in a context-sensitive
analysis framework (without performing
strong updates).
Unlike these flow-insensitive analyses, which are 
not effective for many clients like \upc, \vfsu can perform
strong updates on-demand flow and context-sensitively.

\textsc{Boomerang}~\cite{spath2016boomerang} represents
a recent flow- and context-sensitive demand-driven
pointer analysis for Java.
However, its access-path-based analysis performs only
strong updates partially at a store $a.f=\dots$, by
updating $a.f$ strongly but the aliases of $a.f.*$
weakly, where $a$ and $b$ are different top-level
variables. Let us explain this by using the following straight-line Java code and its corresponding C code. 

\begin{table}[h]
\centering
\vspace{-5mm}
\begin{tabular}{@{\hspace{8mm}}l@{\hspace{10mm}}l}
\\\hline
\begin{tabular}{l}
	\multicolumn{1}{c}{Java Code} \\ \hline
$\lab_1$:\quad q = new\ A()   \hspace*{1.4ex}\quad\textcolor{green}{// o1} \\
$\lab_2$:\quad p = q \\
$\lab_3$:\quad p.f = new A()  \quad\textcolor{green}{// o2} \\
$\lab_4$:\quad q.f = new A() \quad\textcolor{green}{// o3} \\
$\lab_5$:\quad x = p.f
\\\hline
\end{tabular}
&

\begin{tabular}{l}
	\multicolumn{1}{c}{C Code} \\ \hline
$\lab_1$:\quad q = malloc() \hspace*{1.4ex}\quad\textcolor{green}{// o1} \\
$\lab_2$:\quad p = q \\
$\lab_3$:\quad *p = malloc() \quad\textcolor{green}{// o2} \\
$\lab_4$:\quad *q = malloc() \quad\textcolor{green}{// o3} \\
$\lab_5$:\quad x = *p 
\\\hline
\end{tabular}


\\ \hline
\end{tabular}
\end{table}

Let us consider \textsc{Boomerang} first. 
At $\lab_3$, a strong update 
is performed to $p.f$ to make it point to $o2$ only.
At $\lab_4$, a strong update is performed
to $q.f$ to make it point to $o3$ but 
a weak update is performed to all its aliases
so that $p.f$ now points to not only $o2$ as before 
but also $o3$, As a result, $x$ 
points-to both $o2$ and $o3$ at $\lab_5$. 
Let us consider now \supa. 
With both flow- and context-sensitivity enforced, 
a strong update is performed to $o1$ pointed $p$
and $q$ at both $\lab_3$ and $\lab_4$, respectively. Thus, 
$x$ points to $o3$ only at $\lab_5$.

\subsection{Hybrid Pointer Analysis}

The basic idea is to find a right balance between efficiency and 
precision. For C programs,
the one-level approach~\cite{Das:2000} achieves a precision 
between Steensgaard's and Andersen's analyses by applying 
a unification process to address-taken variables only. In the case of
Java programs, context-sensitivity can be made more effective 
by considering both
call-site-sensitivity and object-sensitivity together than either
alone~\cite{kastrinis2013hybrid}.
In~\cite{Guyer:2003}, how to adjust the
analysis precision according to a client's needs is 
discussed.
\citeauthor{Zhang:2014}~\citeyear{Zhang:2014} focus on finding
effective abstractions for whole-program analyses 
written in Datalog via abstraction refinement.
Lhot\'{a}k and Chung \cite{strongupdate} trades precision for
efficiency by performing strong updates only on flow-sensitive
singleton objects but falls back to the flow-insensitive points-to
information otherwise. In this paper, we propose to carry out our
on-demand strong update analysis in a hybrid multi-stage analysis
framework. Unlike \cite{strongupdate}, \vfsu can achieve
the same precision as whole-program flow-sensitive analysis, subject to a given budget.

\subsection{Parallel Pointer Analysis}

\citeauthor{Mendez-Lojo:2010}~\citeyear{Mendez-Lojo:2010} introduced a parallel 
implementation of Andersen's analysis for C based on graph rewriting.
Their parallel analysis is flow- and context-insensitive, achieving
a speedup of up to 3X on an 8-core CPU. \citeauthor{Su15}~\citeyear{Su15}
introduces an improvement
of this parallel implementation on GPUs.
The whole-program sparse flow-sensitive pointer analysis~\cite{hardekopf2009semi} has also been 
parallelized on multi-core CPUs~\cite{Nagaraj:2013} and GPUs~\cite{Nasre:2013:TSF}. The speedups are up to 2.6X on a 8-core CPU. This paper presents the first parallel implementation of demand-driven pointer analysis with strong updates for C programs, achieving an average speedup of 3.05X on a 8-core CPU.

\section{Conclusion}
\label{sec:conc}

We have introduced, \vfsu, a demand-driven pointer
analysis that enables computing precise points-to information for C programs flow- and context-sensitively with strong
updates by
refining away imprecisely pre-computed value-flows, 
subject to some analysis budgets. \vfsu handles large 
C programs effectively
by allowing pointer analyses with different efficiency and precision tradeoffs
to be applied in a hybrid multi-stage analysis
framework. \vfsu
is particularly suitable for environments with small time
and memory budgets such as IDEs. We have evaluated
\vfsu by choosing 
uninitialized pointer detection as a major client on
18 C programs. \vfsu can achieve
nearly the same precision as whole-program
flow-sensitive analysis under small budgets.

One interesting future work is to 
investigate how to allocate budgets in \vfsu to its stages to 
improve the precision achieved in answering some time-consuming queries
for a particular client.
Another direction is to add more stages to its analysis,
by considering, for example, path correlations.

\balance

\bibliographystyle{apalike}
\bibliography{oopsla15}

\end{document}